\documentclass[11pt]{article}
\usepackage[colorlinks=true,linkcolor=red,citecolor=blue]{hyperref}
\usepackage{graphicx}
\usepackage{verbatim}
\usepackage{array,fullpage,multirow}
\usepackage{latexsym}
\usepackage{enumitem}
\usepackage{amssymb,amsfonts,amsmath,amsthm}
\usepackage{url}

\newcounter{maintheoremcounter}

\theoremstyle{plain}
\newtheorem{theorem}{Theorem}[section]

\newtheorem{maintheorem}[maintheoremcounter]{Theorem}

\newtheorem{lemma}[theorem]{Lemma}

\newtheorem{corollary}[theorem]{Corollary}
\newtheorem{claim}[theorem]{Claim}

\newtheorem{definition}[theorem]{Definition}
\newtheorem{fact}[theorem]{Fact}
\newtheorem{openproblem}{Open Problem}

\numberwithin{equation}{section} 

\newcommand{\etal}{\textit{et~al.}}

\def\SD{\ensuremath{\mathrm{SD}}}

\def\F{\ensuremath{\mathbb{F}}}
\def\ftwo{\ensuremath{\mathbb{F}_{2}}}

\def\Exp{\mathbb{E}}
\def\bias{\ensuremath{\mathrm{bias}} }
\def\Pr{\ensuremath{\mathrm{Pr}} }
\def\dim{\ensuremath{\mathrm{dim}} }
\def\spa{\mathrm{span}}
\def\poly{{\mathrm{poly}}}
\def\wt{{\mathrm{wt}}}
\def\polylog{{\mathrm{polylog}}}

\def\deff{\triangleq}

\newcommand{\eps}{\varepsilon}

\newcommand{\XOR}{\mathsf{XOR}}
\newcommand{\OR}{\mathsf{OR}}
\newcommand{\AND}{\mathsf{AND}}

\newcommand{\XAX}{\XOR\!-\!\AND\!-\!\XOR}
\newcommand{\NP}{\bf{NP}}

\newcommand{\size}{\mathrm{size}}

\newcommand{\ACZERO}{\mathsf{AC}^0}
\newcommand{\ACZEROPARITY}{\mathsf{AC}^0[\oplus]}

\def\F{{\mathbb{F}}}

\def\N{{\mathbb{N}}}
\def\Z{{\mathbb{Z}}}

\def\one{\boldsymbol 1}

\renewcommand{\Pr}{\mathop{\bf Pr\/}}
\newcommand{\E}{\mathop{\bf E\/}}

\newcommand{\Variety}{\mathop{\bf V\/}}

\def\poly{{\mathrm{poly}}}
\def\bias{{\mathrm{bias}}}
\def\spn{\mathrm{span}}

\def\({\left(}
\def\){\right)}

\newcommand{\remove}[1]{}

\newcommand{\ConstStruct}{c_{1}}
\newcommand{\ConstStructTwo}{c_{2}}

\begin{document}

\title{Two Structural Results for Low Degree Polynomials \\ and Applications}

\author{Gil Cohen\thanks{
    Department of Computer Science and Applied Mathematics, Weizmann Institute of Science, Rehovot 76100, { \sc Israel}.
    {\tt gil.cohen@weizmann.ac.il.}
    Supported by an ISF grant and
    by the I-CORE Program of the Planning and Budgeting Committee.
}\and Avishay Tal\thanks{
    Department of Computer Science and Applied Mathematics, Weizmann Institute of Science, Rehovot 76100, { \sc Israel}.
    {\tt avishay.tal@weizmann.ac.il.}
    Supported by an Adams Fellowship of the Israel Academy of Sciences and Humanities,
    by an ISF grant and
    by the I-CORE Program of the Planning and Budgeting Committee.
}}

\maketitle

\begin{abstract}
In this paper, two structural results concerning low degree polynomials over finite fields are given. The first states that over any finite field $\mathbb{F}$, for any polynomial $f$ on $n$ variables with degree $d \le \log(n)/10$, there exists a subspace of $\mathbb{F}^n$ with dimension $\Omega(d \cdot n^{1/(d-1)})$ on which $f$ is constant. This result is shown to be tight. Stated differently, a degree $d$ polynomial cannot compute an affine disperser for dimension smaller than $\Omega(d \cdot n^{1/(d-1)})$. Using a recursive argument, we obtain our second structural result, showing that any degree $d$ polynomial $f$ induces a partition of $\mathbb{F}^n$ to affine subspaces of dimension $\Omega(n^{1/(d-1)!})$, such that $f$ is constant on each part.

We extend both structural results to more than one polynomial. We further prove an analog of the first structural result to sparse polynomials (with no restriction on the degree) and to functions that are close to low degree polynomials. We also consider the algorithmic aspect of the two structural results.

\medskip
\noindent
Our structural results have various applications, two of which are:
\begin{itemize}
  \item Dvir [CC 2012] introduced the notion of extractors for varieties, and gave explicit constructions of such extractors over large fields. We show that over any finite field any affine extractor is also an extractor for varieties with related parameters. Our reduction also holds for dispersers, and we conclude that Shaltiel's affine disperser [FOCS 2011] is a disperser for varieties over $\mathbb{F}_2$.

  \item Ben-Sasson and Kopparty [SIAM J. C 2012] proved that any degree $3$ affine disperser over a prime field is also an affine extractor with related parameters. Using our structural results, and based on the work of  Kaufman and Lovett [FOCS 2008] and Haramaty and Shpilka [STOC 2010], we generalize this result to any constant degree.

\end{itemize}
\medskip
\noindent
\end{abstract}

\thispagestyle{empty}
\newpage
\small
\hypersetup{linkcolor=black}
\tableofcontents
\hypersetup{linkcolor=red}
\normalsize
\thispagestyle{empty}
\newpage
\pagenumbering{arabic}

\section{Introduction}

In this paper, we consider the following question concerning polynomials on $n$ variables over the field with $q$ elements, $\F_q$, where $q$ is some prime power:

\begin{quote}
  What is the largest number $k = k_q(n,d)$, such that any polynomial on $n$ variables over $\F_q$, with degree at most $d$, is constant on some affine subspace of $\F_q^n$ with dimension $k$?
\end{quote}
Here, and throughout the paper, by degree we mean total degree.

This question concerning the structure of low degree polynomials over finite fields can be rephrased, in the language of pseudorandomness, as whether a low degree polynomial can be a good affine disperser. Recall that an \emph{affine disperser} for dimension $k$ is a function $f \colon \F_q^n \to \F_q$ with the following property. For every affine subspace $u_0 + U \subseteq \F_q^n$ of dimension $k$, $f$ restricted to $u_0+U$ is not constant~\footnote{An alternative definition requires that almost all field elements are obtained by $f$ on $u_0+U$.}. A function $f \colon \F_q^n \to \F_q$ is called an \emph{affine extractor} for dimension $k$ with bias $\eps$, if for every affine subspace $u_0 + U \subseteq \F_q^n$ of dimension $k$, it holds that $f(x)$, where $x$ is sampled uniformly from $u_0 + U$, is $\eps$-close in statistical distance, to the uniform distribution over $\F_q$.

It is worth mentioning that several explicit constructions of affine dispersers and affine extractors are in fact low degree polynomials~\cite{B07, BSG12, BSK12}. Examples of this fact can be found in the literature for other types of dispersers and extractors as well~\cite{CG88, BIW06,Dvir12}. In fact, the state of the art explicit construction of affine extractors over $\F_2$ by Li~\cite{Li11} (matching the parameters obtained by Yehudayoff~\cite{Y11}) heavily relies on low degree seeded extractors.


Clearly, for any $q$ it holds that $k_q(n,1) = n-1$. The case $d=2$, at least over fields of characteristic $2$, is also well understood. By Dickson's theorem (\cite{D01}, Theorem 199), $k_q(n,2) \le n/2 + 1$ for fields of characteristic $2$. This is tight, as can be seen by considering the inner product function $x_1 x_2 + x_3 x_4 + \cdots + x_{n-1} x_n$. To the best of our knowledge, the value of $k_q(n,d)$ has not received a formal treatment in the literature, and in particular it is not known or can be easily deduced by previous works, for $d > 2$. The most related result was obtained by Barrington and Tardos (\cite{TB98}, Lemma 3), who proved that for any prime power $q$ and for any degree $d$ polynomial $f$ on $n$ variables over the \emph{ring} $\Z_q$, there exists a ``cube'' with dimension $k = \Omega(n^{1/d})$, on which $f$ is constant. That is, there exist linearly independent vectors $\Delta_1, \ldots, \Delta_k \in \Z_q^n$ such that for every $\alpha \in \{0,1\}^k$, $f(\sum_{i=1}^{k}{\alpha_i \Delta_i}) = f(0)$.



Furthermore, a natural variant of the question of understanding $k_q(n,d)$ was previously raised by Trevisan~\cite{T06} for the special case $q = 2$. As a corollary of the structural results of Haramaty and Shpilka~\cite{HS10}, for biased polynomials and for polynomials with large Gowers norm with degree $d=3,4$ over prime fields, one can deduce non-trivial lower bounds on the dimension of an affine subspace on which such polynomials are constant. Assuming low degree and bounded spectral norm, lower bounds on the affine subspace dimension follow by the structural result of Tsang~\etal~\cite{TWXZ13}.

\subsection{Our Results}

The first result of this paper is an asymptotically tight upper and lower bounds on $k_q(n,d)$ for all $d < \log(n)/10$. Our lower bound holds for all finite fields, namely, for any prime power $q$. We then further study the structure of low degree polynomials over finite fields, the algorithmic aspect of these results, and present several applications to complexity theory and in particular to pseudorandomness.

The following theorem gives a lower bound for $k_q(n,d)$. In fact, it promises something stronger, which is required by one of our applications (see Theorem~\ref{thm:from affine disperser to varieties for small t}). Informally, for any degree $d$ polynomial $f$ and a point $u_0 \in \F_q^n$, there exists a large subspace $U$ such that $f$ is constant on $u_0 + U$. Note that this is equivalent of saying that there exists a large subspace on which $f$ is constant.

\begin{maintheorem}[Structural Result I]\label{thm:structural result}
For any $n, d$, let $k$ be the least integer such that
\begin{equation}\label{eq:bound}
n \le k + (d+1) \cdot \sum_{j=0}^{d-1} {(d-j) \cdot \binom{k+j-1}{j}}\;.
\end{equation}
Let $q$ be a prime power. Let $f \colon \F_q^n \to \F_q$ be a degree $d$ polynomial, and let $u_0 \in \F_q^n$. Then, there exists a subspace $U \subseteq \F_q^n$ of dimension $k$ such that $f|_{u_0+U}$ is constant.

In particular, there exists a universal constant $\ConstStruct \in (0,1)$ such that for all $n,d,q$, it holds that $k_q(n,d) \ge \ConstStruct \cdot n^{1/(d-1)}$. Moreover, for $d \le \log(n)/10$ it holds that $k_q(n,d) = \Omega(d \cdot n^{1/(d-1)})$.
\end{maintheorem}

Few remarks are in order. First, we note that Theorem~\ref{thm:structural result} is tight for $d\le \log(n)/10$. Indeed, one can show that, with probability at most $q^{-\binom{k}{d}}$~\footnote{The expression $\binom{k}{d}$ in the exponent can be replaced by the number of solutions to the equation $r_1 + \ldots + r_k \le d$, where $r_i \in \{0,\ldots, q-1\}$.}, a random degree $d$ polynomial on $n$ variables over $\F_q$ is constant on any fixed affine subspace of dimension $k$. There are at most $q^{(k+1)n}$ affine subspaces of dimension $k$, so by the union bound, $k_q(n,d)$ must be smaller than any $k$ such that $\binom{k}{d} > (k+1)n$. Hence, $k_q(n,d) < d^{1+1/(d-1)} \cdot n^{1/(d-1)}$. For $d \le \log(n)/10$, the ratio between our upper and lower bound is $d^{O(1/d)} = 1 + O(\log(d)/d)$.

For the special case $q=2$, based on the work of Ben-Eliezer~\etal~\cite{BHL09}, one can say something stronger regarding the tightness of Theorem~\ref{thm:structural result}. Namely, for every $d \ge 1$, there exists a degree $d$ polynomial $f \colon \F_2^n \to \F_2$ that has bias $2^{-\Omega(k/d)}$ on any affine subspace of dimension $k \ge \Omega(d \cdot n^{1/(d-1)})$ (see Section~\ref{sec:tightness}). In the language of pseudorandomness, Theorem~\ref{thm:structural result} states that a degree $d \le \log{(n)}/10$ polynomial is not an affine disperser for dimension $o(d \cdot n^{1/(d-1)})$, and in particular, polynomials with constant degree are not affine dispersers for sub-polynomial dimension. The tightness results mentioned above, imply that there exists a degree $d$ polynomial which is an affine disperser for dimension $k = O(d \cdot n^{1/(d-1)})$, over any finite field. Moreover, for the special case $q=2$, there exists a degree $d$ polynomial that is an affine extractor for the same dimension $k = O(d \cdot n^{1/(d-1)})$, with bias $2^{-\Omega(k/d)}$.

While the results of Barrington and Tardos~\cite{TB98} concern the ring $\Z_q$, where $q$ is a prime power, our results concern the field $\F_q$, making the results incomparable in that sense. However, \cite{TB98} guarantees the existence of a cube (as defined above), which is weaker than the existence of an affine subspace guaranteed by Theorem~\ref{thm:structural result}. These two notions are equivalent only for the special case $q = 2$. Furthermore, the dimension of the affine subspace obtained by Theorem~\ref{thm:structural result} is $\Omega(n^{1/(d-1)})$, which is larger than $\Omega(n^{1/d})$ -- the dimension of the cube obtained by Barrington and Tardos. Although this difference may seem small, it is crucial for one of our applications concerning a reduction from affine extractors to affine dispersers (see Section~\ref{sec:reduction from extractors to dispersers}). On the other hand, the cube obtained by Barrington and Tardos has a structure that is necessary for their application (the latter concerns the minimum degree of a polynomial over rings representing the $\OR$ function), namely, the vectors $\Delta_1, \ldots, \Delta_k$ are not only linearly independent over $\Z_q$, but in fact have disjoint supports.

Note that the bound on $k_q(n,d)$ in Theorem~\ref{thm:structural result} is independent of $q$. That is, when considering bounded degree polynomials, the field size does not affect $k_q(n,d)$. To be more precise, one can replace the term $(d+1)$ that multiplies the sum in Equation~\eqref{eq:bound} with $\min(d+1,q)$. In any case, the term $\min(d+1,q)$ has no affect over $k_q(n,d)$ since the $d-1$ root is taken to isolate $k$ in the equation. Throughout the paper we focus on low degree polynomials -- polynomials of degree up to $\log(n)/10$. In this range of parameters, Theorem~\ref{thm:structural result} and the fact that it is tight, allow us to suppress the field size and write $k(n,d)$ instead of $k_q(n,d)$, as we do from here on.

When the degree of the polynomial is unbounded (which  boils down to the question of understanding the parameters of optimal affine dispersers), things behave differently. In other words, by increasing the field size, one can obtain affine dispersers for smaller dimension. For example, it is known that any function $f : \F_2^n \to \F_2$ is constant on some affine subspace with dimension $\Omega(\log{n})$. Namely, $k_2(n,\infty) = \Omega(\log{n})$ (this is, in fact, tight). On the other hand, Gabizon and Raz~\cite{GR08} noted that the polynomial $x_1^1 + x_2^2 + \cdots + x_n^n$ over the field with $n+1$ elements is not constant on any dimension $1$ affine subspace (see also~\cite{DG10}). Thus, $k_{n+1}(n,\infty) = 1$. Understanding the correct value of $k_3(n,\infty)$ seems to be an interesting open problem.

\paragraph{Partition of $\F^n$ to affine subspaces, induced by a low degree polynomial.}

\medskip
Theorem~\ref{thm:structural result} states that for any degree $d$ polynomial $f$ on $n$ variables, there exists at least one large affine subspace, restricted to which, $f$ is constant. However, for some of our applications we need a stronger structural result. More specifically, we ask what is the maximum number $\mathcal{K} = \mathcal{K}_q(n,d)$, such that any degree $d$ polynomial on $n$ variables over $\F_q$, induces a \emph{partition} of $\F_q^n$ to dimension $\mathcal{K}$ affine subspaces, on each of which $f$ is constant. Using Theorem~\ref{thm:structural result}, we show that $\mathcal{K}_q(n,d) = \Omega(n^{1/(d-1)!})$. That is, we obtain the following result.

\begin{maintheorem}[Structural Result II]\label{thm:second structural result over Fq}
There exists a universal constant $\ConstStructTwo > 0$ such that the following holds. Let $q$ be a prime power. Let $f \colon \F_q^n \to \F_q$ be a degree $d$ polynomial. Then, there exists a partition of $\F_q^n$ to affine subspaces (not necessarily shifts of the same subspace), each of dimension $\ConstStructTwo \cdot n^{1/(d-1)!}$, such that $f$ is constant on each part.
\end{maintheorem}
We do not know whether the lower bound in Theorem~\ref{thm:second structural result over Fq} for $\mathcal{K}_q(n,d)$ is tight or not for all $d$ (note that it is tight for $d \le 3$), and leave this as an open problem.

\begin{openproblem}
What is the asymptotic behavior of $\mathcal{K}_q(n,d)$? Does it depend on $q$ for, say, constant $d$?
\end{openproblem}

\paragraph{Generalization of the structural results to many polynomials.}
Being a natural generalization and also necessary for some of our applications, we generalize the two structural results to the case of any number of polynomials (see Section~\ref{sec:many poly}). Let $f_1, \ldots, f_t \colon \F_q^n \to \F_q$ be polynomials of degree at most $d$. The generalization of the first structural result states that there exists an affine subspace of dimension $\Omega((n/t)^{1/(d-1)})$ on which \emph{each} of the $t$ polynomials is constant (see Theorem~\ref{thm:structural result I for many polynomials}). By applying a probabilistic argument, one can show that the dependency in $t$ is tight. For the second structural result, the promised dimension in Theorem~\ref{thm:second structural result over Fq} is replaced by $\Omega(n^{1/(d-1)!} / t^e)$, where $e$ is the base of the natural logarithm (see Theorem~\ref{thm:structural result II for many polynomials}).

\paragraph{The algorithmic aspect.}
We further study the algorithmic aspect of the structural results (see Section~\ref{sec:alg}). We devise a $\poly(n)$-time algorithm (see Theorem~\ref{thm:efficient structural result I}), that given a degree $d$ polynomial $f \colon \F_2^n \to \F_2$ as a black-box, performs $\poly(n)$ queries, and outputs a subspace of dimension $\Omega(k(n,d))$, restricted to which, $f$ has degree at most $d-1$. By applying this algorithm recursively $d$ times, one can efficiently obtain a subspace of dimension $\Omega(n^{1/(d-1)!})$ on which $f$ is constant. Our algorithm only works for the binary field. Devising an algorithm for general fields is a natural problem.

Note that there is a gap between $k(n,d)$ and the dimension of the affine subspace that our algorithm produce. A natural open problem is whether  this gap can be eliminated.

\begin{openproblem}\label{op:algorithm}
Is there a $\poly(n)$-time algorithm that, given a black-box access to a degree $d$ polynomial $f \colon \F_2^n \to \F_2$, finds an affine subspace with dimension $k(n,d)$ on which $f$ is constant ?
\end{openproblem}

Whether there exists an algorithm as in Problem~\ref{op:algorithm} is not at all clear to us. Verifying that a degree $d$ polynomial is constant on a given affine subspace with dimension $k(n,d)$ can be done in time $O(k(n,d)^d) \le O(n^2)$, and it might be the case that this problem is expressive enough to be $\NP$-hard. We show that the latter scenario is unlikely, at least for constant $d$, by devising an $\exp(n^{1 - \frac1{d-1}}) \cdot n^d$-time algorithm that outputs an affine subspace with dimension $\Omega(k(n,d))$ on which $f$ is constant (see Theorem~\ref{thm:subexp alg}). We note that the naive algorithm iterates over all $\binom{2^n}{k(n,d)} = \exp(n^{1+\frac1{d-1}})$ affine subspaces with dimension $k(n,d)$. It is also worth mentioning that this algorithm works for all finite fields.

\paragraph{Sparse polynomials.}

We further give an analog of the first structural result to sparse polynomials (regardless of their degree) over any finite field. We have the following.

\begin{maintheorem}\label{thm:sparse}
Let $q$ be a prime power. For any integer $c \ge 1$ the following holds. Let $f$ be a polynomial on $n$ variables over $\F_q$, with at most $n^c$ monomials. Then, there exists an affine subspace of dimension $\Omega\left(n^{1/(4(q-1)c)}\right)$ on which $f$ is constant.
\end{maintheorem}

We note that unlike in the case of low degree polynomials, the field size $q$ does affect the dimension of the affine subspace promised by Theorem~\ref{thm:sparse}. Some sort of dependency cannot be avoided. Indeed, as mentioned above, the polynomial $x_1^1 + x_2^2 + \cdots + x_n^n$ over the field with $n+1$ elements is not constant on any dimension $1$ affine subspace, even though it has only $n$ monomials. On the other hand, Theorem~\ref{thm:sparse} gives no guarantee already for $q = \Omega(\log{n})$, while the example above requires fields of size $\Omega(n)$. We leave open the problem of improving upon the dependency of Theorem~\ref{thm:sparse} in the field size $q$, or proving that this dependency is optimal.

\begin{openproblem}
What is the correct dependency in the field size $q$ for the class of sparse polynomials ?
\end{openproblem}

We note that for the special case $q = 2$, the lower bound in Theorem~\ref{thm:sparse} is $\Omega\left(n^{1/(4c)}\right)$, which is essentially tight up to the constant $4$ in the exponent, as implied by our tightness result for degree $d$ polynomials. We do not know whether the constant $4$ is necessary. Indeed, for degree $d$ polynomials (which may have $n^d$ monomials), the guarantee given by Theorem~\ref{thm:structural result} is stronger, namely, $\Omega\left(n^{1/(d-1)}\right)$.

\paragraph{Functions that are close to low degree polynomials.}

Theorem~\ref{thm:structural result} implies that any function that is close to a low degree polynomial, is constant on some large affine subspace.

\begin{corollary}\label{cor:approx}
Let $q$ be a prime power. Let $g : \F_q^n \to \F_q$ be a function that agrees with some degree $d$ polynomial $f : \F_q^n \to \F_q$ on all points but for some subset $B \subseteq \F_q^n$. Then, there exists an affine subspace with dimension
$
 \Omega( (n - \log_q(|B|))^{1/(d-1)})
$
on which $g$ is constant.
\end{corollary}
To see that, note that by averaging argument there is an affine subspace $w + W$ of dimension $n-\log_q(|B|)-1$ on which $f$ and $g$ agrees.
Applying Theorem~\ref{thm:structural result} to $f|_{w+W}$ gives an affine subspace $u+U \subseteq w + W$ on which $f$, and thus $g$, is constant on.  We suspect that better parameters can be achieved.
%
%

\subsection{Applications}

We now present several applications of our structural results. 

\subsubsection*{Extractors and Dispersers for Varieties over all Finite Fields}

Let $\F$ be some finite field. An affine subspace of $\F^n$ can be thought of as the set of common zeros of one or more degree 1 polynomials with coefficients in $\F$. Recall that an affine extractor over the field $\F$ is a function $f \colon \F^n \to \F$ that has small bias on every large enough affine subspace. In \cite{Dvir12}, the study of the following natural generalization was initiated: construct a function that has small bias on the set of common zeros of one or more degree $d>1$ polynomials. In general, the set of common zeros of one or more polynomials is called a \emph{variety}. For a set of polynomials $g_1, \ldots, g_t$ on $n$ variables over $\F$, we denote their variety by
$$
\Variety(g_1, \ldots, g_t) = \left\{ x \in \F^n : g_1(x) = \cdots = g_t(x) = 0 \right\}.
$$
A function $f \colon \F^n \to \F$ as above is called an extractor for varieties.

In \cite{Dvir12}, two explicit constructions of extractors for varieties were given. For simplicity, we suppress here both the bias of the extractor and the number of output bits. Dvir's first construction works under no assumption on the variety size (more precisely, some assumption is made, but that assumption is necessary). The downside of this construction is that the underlining field is assumed to be quite large, more precisely, $|\F| > d^{\Omega(n^2)}$. The second construction works for fields with size as small as $\poly(d)$, however the construction is promised to work only for varieties with size at least $|\F|^{n/2}$. Dvir applies tools from algebraic geometry for his constructions.

Even the construction of affine extractors, which is a special case of extractors for varieties, is extremely challenging. Indeed, the (far from optimal) constructions known today use either very sophisticated exponential sum estimates~\cite{B07,Y11} or involved composition techniques~\cite{Li11}, where the correctness relies, among other results, on deep structural results from additive combinatorics~\cite{Vin11} and on XOR lemmas for low degree polynomials~\cite{VW07, BKSSZ10}. The same can be said about the constructions of affine dispersers.

Given the difficulties in constructing affine extractors and dispersers, one may suspect that the construction of extractors and dispersers for varieties will be substantially more challenging, especially for small fields that seem to be immune against algebraic geometry based techniques. Nevertheless, based on our structural results, the following theorem states that any affine extractor is also an extractor for varieties with related parameters.


\begin{maintheorem}\label{thm:from affine extractors to varieties for small t intro version}
Let $q$ be a prime power. For any integers $n,d,t$ the following holds.  Let $f \colon \F_q^n \to \F_q$ be an affine extractor for dimension $\Omega(n^{1/(d-1)!} / t^e)$ with bias $\eps$. Then, $f$ is an extractor with bias $\eps$ for varieties that are the common zeros of any $t$ polynomials, each of degree at most $d$.
\end{maintheorem}
In fact, one can view Theorem~\ref{thm:from affine extractors to varieties for small t intro version} as an explanation for the difficulty of constructing affine extractors for dimension $n^\delta$ for constant $\delta < 1$. 

We also obtain a reduction that does not depend on the number of polynomials defining the variety, but rather on the variety size (see Theorem~\ref{thm:from affine extractors to varieties for any t}). The proof idea in this case is to ``approximate'' the given variety by a variety induced by a small number of low degree polynomials, and then apply Theorem~\ref{thm:from affine extractors to varieties for small t intro version}.

The state of the art explicit constructions of affine extractors for the extreme case $q = 2$, work only for dimension $\Omega(n/\sqrt{\log{\log{n}}})$ \cite{B07,Y11,Li11}, and thus the reduction in Theorem~\ref{thm:from affine extractors to varieties for small t intro version} only gives an explicit construction of an extractor for varieties defined by quadratic polynomials (and in fact, up to $(\log{\log{n}})^{1/(2e)}$ quadratic polynomials). However, a similar reduction to that in Theorem~\ref{thm:from affine extractors to varieties for small t intro version} also holds for dispersers.

\begin{maintheorem}\label{thm:from affine disperser to varieties for small t}
Let $n,d,t$ be integers such that $d < \log(n/t)/10$. Let $f \colon \F_q^n \to \F_q$ be an affine disperser for dimension $\Omega(d \cdot (n/t)^{1/(d-1)})$. Then, $f$ is a disperser for varieties that are the common zeros of any $t$ polynomials of degree at most $d$.
\end{maintheorem}
Over $\F_2$, an explicit construction of an affine disperser for dimension as small as $2^{\log^{0.9}{n}}$ is known~\cite{S11}. Thus, we obtain the first disperser for varieties over $\F_2$.

\begin{maintheorem}\label{thm:explicit disperser intro version}
For any $n,d,t$ such that $d < (1-o_n(1)) \cdot \frac{\log{(n/t)}}{\log^{0.9}{n}}$, there exists an explicit construction of an affine disperser for varieties which are the common zeros of any $t$ polynomials of degree at most $d$. In particular, when $t \le n^\alpha$ for some constant $\alpha < 1$, the requirement on the degree is $d < (1-\alpha - o_n(1)) \cdot \log^{0.1}{n}$.
\end{maintheorem}

A few words regarding the limitation of the reduction in Theorem~\ref{thm:from affine disperser to varieties for small t} are in order. Note that even if $f$ is an optimal affine disperser, that is, a disperser for dimension $O(\log{n})$, Theorem~\ref{thm:from affine disperser to varieties for small t} only guarantees that $f$ is a disperser for varieties defined by degree $O(\log{n})$ polynomials. One cannot expect much more from the reduction. Indeed, there exists a degree $O(\log{n})$ polynomial that computes an optimal affine disperser (this can be proven via a probabilistic argument. See also Theorem~\ref{thm:affine construction}). However, this affine disperser is clearly not a disperser for varieties defined by even a single degree $O(\log{n})$ polynomial.

Thus, the reduction in Theorem~\ref{thm:from affine disperser to varieties for small t} is useful only for varieties defined by degree $o(\log{n})$ polynomials. A recent work of Hrube\v{s} and Rao~\cite{HR14} shows that it would be challenging to construct an explicit $f$ which is an extractor (or even a disperser) for varieties of size $2^{\rho n}$ defined by degree $n^{\eps}$ polynomials over $\F_2$, for any constants $0 < \eps, \rho<1$. Indeed, such a function would solve Valiant's problem \cite{Val77}, since $f$ cannot be computed by Boolean circuits of logarithmic depth and linear size.

\subsubsection*{From Affine Dispersers to Affine Extractors}
Constructing an affine disperser is, by definition, an easier task than constructing an affine extractor. Nevertheless, Ben-Sasson and Kopparty~\cite{BSK12} proved (among other results) that any degree $3$ affine disperser is also an affine extractor with comparable parameters.~\footnote{A reduction from ``low rank'' extractors to dispersers in the context of two sources was also obtained, by Ben-Sasson and Zewi~\cite{BZ11}, conditioned on the well-known Polynomial Freiman-Ruzsa conjecture from additive combinatorics.}
Using the extension of Theorem~\ref{thm:structural result} to many polynomials, we are able to generalize the reduction of Ben-Sasson and Kopparty, over prime fields, to any degree $d \ge 3$.

\begin{maintheorem}\label{thm:reduction from extractors to dispersers intro version}
Let $p$ be a prime number. For all $d \ge 3$ and $\delta > 0$, there exists $c = c(d,\delta)$ such that the following holds.
Let $f \colon \F_p^n \to \F_p$ be an affine disperser for dimension $k$, which has degree $d$ as a polynomial over $\F_p$.
Then, $f$ is also an affine extractor for dimension $k' \deff c \cdot k^{d-2}$ with bias $\delta$.
\end{maintheorem}
Note that Theorem~\ref{thm:reduction from extractors to dispersers intro version} is only interesting in the case where $k^{d-2} < n$.
However, this case is achievable since a random polynomial of degree $d$ is an affine disperser for dimension $O(d \cdot n^{1/(d-1)})$. In particular, Theorem~\ref{thm:reduction from extractors to dispersers intro version} implies that an explicit construction of an optimal affine disperser that has a constant degree as a polynomial, suffices to break the current natural barrier in the construction of affine extractors, namely, constructing affine extractors for dimension $n^{1-\delta}$ for some constant $\delta>0$ (here $\delta = 1/(d-1)$).

On top of Theorem~\ref{thm:structural result}, the key ingredient we use in the proof of Theorem~\ref{thm:reduction from extractors to dispersers intro version} is the work of Kaufman and Lovett~\cite{KL08}, generalizing a result by Green and Tao~\cite{GT09} (see Section~\ref{sec:reduction from extractors to dispersers}).
For $d=4$, we get a better dependency between $k$ and $k'$ based on the work of Haramaty and Shpilka \cite{HS10} (see Theorem~\ref{thm:degree 3 4}).



\subsubsection*{$\boldsymbol{\ACZEROPARITY}$ Circuits and Affine Extractors / Dispersers}\label{intro:circuits affine extractors dispersers}

Constructing affine dispersers, and especially affine extractors, is a challenging task. As mentioned, the state of the art explicit constructions for affine extractors over $\F_2$ work only for dimension $\Omega(n/\sqrt{\log{\log{n}}})$.
By a probabilistic argument however, one can show the existence of affine extractors for dimension $(1+o(1)) \log{n}$ (see Claim~\ref{claim:affine extractors}). Thus, there is an exponential gap between the non-explicit construction and the explicit ones.

It is therefore tempting to try and utilize this situation and prove circuit lower bounds for affine extractors. This idea works smoothly for $\ACZERO$ circuits. Indeed, by applying the work of H{\aa}stad~\cite{Hastad86}, one can easily show that an $\ACZERO$ circuit on $n$ inputs cannot compute an affine disperser for dimension $o(n/\polylog(n))$ (see Corollary~\ref{cor:ACZERO}). However, strong lower bounds for $\ACZERO$ circuits are known, even for much simpler and more explicit functions such as Parity and Majority. Thus, it is far more interesting to prove lower bounds against circuit families for which the known lower bounds are modest. One example would be to show that a De Morgan  formula of size $O(n^3)$ cannot compute a good affine extractor, improving upon the best known lower bound~\cite{H98}.~\footnote{The property of being an affine extractor meets the largeness condition of the natural proof barrier~\cite{RR94}. However, it does not necessarily get in the way of improving existing polynomial lower bounds.}

Somewhat surprisingly, we show that even depth $3$ $\ACZEROPARITY$ circuit (that is, $\ACZERO$ circuits with $\XOR$ gates) can compute an optimal affine extractor over $\F_2$. In fact, the same construction can also be realized by a polynomial-size De Morgan formula and has degree $(1+o(1)) \log{n}$ as polynomial over $\F_2$ (see Theorem~\ref{thm:affine construction}).


Theorem~\ref{thm:affine construction} is implicit in the works of~\cite{Razborov88, Savicky95} who studied a similar problem in the context of bipartite Ramsey graphs (that is, two-source dispersers). We give an alternative proof in Appendix~\ref{sec:depth3 can be affine dispersers}, which can be extended to work also in the context of bipartite Ramsey graphs.

Given that depth $3$ $\ACZEROPARITY$ circuits exhibit the surprising computational power mentioned above, it is natural to ask whether depth $2$ $\ACZEROPARITY$ circuit can compute a good affine extractor. We stress that even depth $2$ $\ACZEROPARITY$ circuits should not be disregarded easily! For example, such circuits \emph{can} compute, in a somewhat different setting, optimal Ramsey graphs (see~\cite{J12}, Section 11.7). Moreover, any degree $d$ polynomial $f \colon \F_2^n \to \F_2$ can be computed by a depth $2$ $\ACZEROPARITY$ circuit with size $n^d$. Nevertheless, we complement the above result by showing that a depth $2$ $\ACZEROPARITY$ circuit cannot compute an affine disperser for sub-polynomial dimension. The proof is based on the following reduction.

\begin{lemma}\label{lemma:reduction from depth 2 to polynomials intro version}
Let $C$ be a depth $2$ $\ACZEROPARITY$ circuit on $n$ inputs, with size $n^c$. Let $k < n/10 - c\log(n)$. If $C$ computes an affine disperser for dimension $k$, then there exists a degree $2c$ polynomial over $\F_2$ on $\sqrt{n} / 5$ variables which is an affine disperser for dimension $k$.
\end{lemma}

The proof of Lemma~\ref{lemma:reduction from depth 2 to polynomials intro version} uses ideas from our proof of the structural result for sparse polynomials (see Lemma~\ref{lem:reduction from sparse}). Lemma~\ref{lemma:reduction from depth 2 to polynomials intro version} together with Theorem~\ref{thm:structural result} imply the following theorem.

\begin{maintheorem}\label{thm:depth 2 are not dispersers intro version}
Let $C$ be a depth $2$ $\ACZEROPARITY$ circuit on $n$ inputs, with size $n^c$, which is an affine disperser for dimension $k$. Then, $k  > k(\sqrt{n}/5,2c) = \Omega(n^{1/4c})$.
\end{maintheorem}
\noindent

\subsubsection*{Good Affine Extractors are Hard to Approximate by Low Degree Polynomials}

Using our second structural result, Theorem~\ref{thm:second structural result over Fq}, we obtain an average-case hardness result, or in other words, correlation bounds for low degree polynomials. Namely, we show that any affine extractor with very good parameters cannot be approximated by low degree polynomials over $\F_2$.

\begin{corollary}\label{cor:average_case intro}
Let $f \colon \F_2^n \to \F_2$ be an affine extractor for dimension $k$ with bias $\eps$.
Then, for any polynomial $g \colon \F_2^n \to \F_2$ of degree $d$ such that $k = \Omega(n^{1/(d-1)!})$, it holds that
\[
\mathrm{Cor}(f, g) \deff \E_{x \sim \F_2^n}{\left[ (-1)^{f(x)} \cdot (-1)^{g(x)} \right]} \le \eps.
\]
\end{corollary}
\begin{proof}
Let $g$ be a degree $d$ polynomial over $\F_2$ on $n$ variables.
By Theorem~\ref{thm:second structural result over Fq}, there exists a partition of $\F_2^n$ to affine subspaces $P_1, P_2,
\ldots, P_\ell$, each of dimension $k = \Omega(n^{1/(d-1)!})$, such that for all $i\in[\ell]$, $g|_{P_i}$ is some constant $g(P_i)$. Thus,
\[
\mathrm{Cor}(f,g) =
\left|\E_{x \sim \F_2^n}[(-1)^{f(x)+g(x)}]\right| =
\left|\E_{i\sim[\ell]} \E_{x\sim P_i} [(-1)^{f(x)+g(P_i)}]\right| \le
\E_{i\sim[\ell]} \left| (-1)^{g(P_i)} \cdot \E_{x\sim P_i} [(-1)^{f(x)}]\right|\;,
\]
which is at most $\eps$ since $f$ is an affine extractor for dimension $k$ with bias $\eps$.
\end{proof}

As mentioned, explicit constructions of affine extractors for dimension $\Omega(n/\sqrt{\log{\log{n}}})$ are known. Corollary~\ref{cor:average_case intro} implies that these extractors cannot be approximated by quadratic polynomials. Corollary~\ref{cor:average_case intro} also implies that for any constant $\beta \in (0,1)$, affine extractors for dimension $k \le 2^{(\log n)^{\beta}}$ with bias $\eps$ have correlation $\eps$ with degree $d \le O_\beta\left( \log {\log {n}} / \log {\log{\log{n}}}\right)$ polynomials.~\footnote{This is the best $d$ we can guarantee for any $k$, and we gain nothing more by taking $k = O(\log{n})$.} Unfortunately, an explicit construction for extractors with such parameters has not yet been achieved.


We also note that stronger correlation bounds are known in the literature for explicit (and simple) functions (see~\cite{V09} and references therein). Nevertheless, we find the fact that \emph{any} affine extractor has small correlation with low degree polynomials interesting.

\subsubsection*{The Granularity of the Fourier Spectrum of Low-Degree Polynomials over $\F_2$}

The bias of an arbitrary function $f \colon \F_2^n \to \F_2$ is clearly some integer multiplication of $2^{-n}$. Theorem~\ref{thm:second structural result over Fq} readily implies that the bias of a degree $d$ polynomial on $n$ variables has a somewhat larger granularity -- the bias is a multiplication of $2^{\Omega(n^{1/(d-1)!})}/2^n$ by some integer.~\footnote{Throughout the paper, for readability, we supress flooring and ceiling. In the last expression, however, it should be noted that we mean $2^{k-n}$, where $k$ is some integer such that $k = \Omega(n^{1/(d-1)!})$.} In fact, Theorem~\ref{thm:second structural result over Fq} implies that \emph{all} Fourier coefficients of a low degree polynomial has this granularity. To see this, apply Theorem~\ref{thm:second structural result over Fq} to obtain a partition $P_1,\ldots,P_\ell$ of $\F_2^n$ to affine subspaces of dimension $k = \Omega(n^{1/(d-1)!})$, such that for each $i \in [\ell]$, $f|_{P_i}$ is some constant $f(P_i)$. Let $\beta \in \F_2^n$. Then,
$$
2^n \cdot \widehat{f}(\beta) = \sum_{x \in \F_2^n}{(-1)^{\left< \beta, x \right>} \cdot (-1)^{f(x)} }
= \sum_{i=1}^{\ell}{\sum_{x \in P_i}{(-1)^{\left< \beta, x \right>} \cdot (-1)^{f(x)}}}
= \sum_{i=1}^{\ell}{(-1)^{f(P_i)} \cdot \sum_{x \in P_i}{(-1)^{\left< \beta, x \right>}}}.
$$
The proof then follows as for all $i \in [\ell]$, the inner sum $\sum_{x \in P_i}{(-1)^{\left< \beta, x \right>}}$ is either $0$ or $\pm 2^k$.

\subsection{Proof Overview}
In this section we give proof sketches for some of our structural results. We start with Theorem~\ref{thm:structural result}, and for simplicity, consider first the special case $q = 2$. In fact, in Appendix~\ref{sec:binary field} we give a full proof for this special case, as it is slightly simpler than the proof for the general case, and conveys some of the ideas used in the proof for the more general case. Our proof is rather elementary, in spite of what one should expect considering previous works in this area, which apply machinery from additive combinatorics and Fourier analysis.

We are given a point $u_0 \in \F_2^n$ and assume, without loss of generality, that $f(u_0) = 0$. We iteratively construct affine subspaces, restricted to which, $f$ is zero.
We start with affine subspaces of dimension $0$, which are just the singletons $\{x\}$, where $x\in \F_2^n$ is such that $f(x)=0$.
Assume that we were able to find basis vectors $\Delta_1, \ldots, \Delta_k$ for a subspace $U$ such that $f$ restricted $u_0 + U$ is constantly $0$. Consider all cosets $x + U$, restricted to which $f$ is constantly $0$. We call such cosets \emph{good}. Clearly the coset $u_0 + U$ is good.
If at least one more good coset $x + U$ exists, then we can pick a new direction $\Delta_{k+1}$ to be $x+u_0$, and get that $f$ is zero on $u_0 + \spa\{\Delta_1, \ldots, \Delta_{k+1}\}$, as indeed
\begin{align*}
u_0 + \spa\{\Delta_1, \ldots, \Delta_{k+1}\} &=
\left( u_0 + \spa\{\Delta_1, \ldots, \Delta_{k}\} \right) \cup \left( u_0 + \Delta_{k+1} + \spa\{\Delta_1, \ldots, \Delta_{k}\} \right) \\ &=
\left( u_0 + \spa\{\Delta_1, \ldots, \Delta_{k}\} \right) \cup \left( x + \spa\{\Delta_1, \ldots, \Delta_{k}\} \right).
\end{align*}

The main observation that allows us to derive Theorem~\ref{thm:structural result} is the following. Given $\Delta_1, \ldots, \Delta_k$, there exists a degree $D \le d^2 \cdot k^{d-1}$ polynomial $t : \F_2^n \to \F_2$, such that $x+U$ is a good coset if and only if $t(x) = 1$. Since we know that $t$ is not the constant $0$ function (as $t(u_0) = 1$), the DeMillo-Lipton-Schwartz-Zippel lemma (see Lemma~\ref{lem:dlsz}) implies that there are at least $2^{n-D}$ $x$'s such that $t(x) =1$, namely, $2^{n-D}$ good cosets.  So in each iteration, by our choice of $\Delta_{k+1}$, we ensure that one coset in the next iteration is good, and then use DeMillo-Lipton-Schwartz-Zippel to claim that many other cosets are good as well. We can continue expanding our subspace $U$ until $n \le D$, which completes the proof.
\medskip

For a general finite field, $\F_q$, we similarly define a polynomial $t(x)$ over $\F_q$ that attains only the values $0$ and $1$, and whose $1$'s capture the good cosets. The polynomial $t(x)$ is of degree at most $q \cdot d^2 \cdot k^{d-1}$.
We wish to find a new direction $\Delta_{k+1}$, linearly independent of $\Delta_1, \ldots, \Delta_k$, such that all cosets along the line $\{u_0 + \Delta_{k+1} \cdot a\}_{a\in\F_q}$, i.e. $\{u_0 + \Delta_{k+1}\cdot a + U \}_{a\in \F_q}$, are good. Over $\F_2$ this task was easy since $u_0 + U$ and $x+ U$ define such a line. The main new idea needed over $\F_q$ is to consider a polynomial
\[
s(y) = \prod_{a \in \F_q} {t(u_0 + y \cdot a)}.
\]
Note that $s(y)$ has degree at most $q \cdot \deg(t)$ and that $s(y) = 1$ if and only if $t(u_0 + y \cdot a) = 1$ for all $a\in \F_q$. 
Thus, $s(y)=1$ iff $f$ is zero on all cosets $\{u_0 + y \cdot a + U\}_{a\in \F_q}$, whose union is a dimension $k+1$ affine subspace as long as $y \notin U$.
As before, since $s(0) = 1$, by a generalized DeMillo-Lipton-Schwartz-Zippel lemma, it holds that $s(\cdot)$ has many $1$'s, and as long as $k \ll n^{1/(d-1)}$ there is some $y \in s^{-1}(1)$  such that $y \notin U$. We can now pick such a $y$ as $\Delta_{k+1}$.
A slightly more careful argument shows that actually there is no dependency of the dimension $k$ in the field size $q$.


\remove{The proof above relies on properties of polynomial derivatives. To extend this proof to non-prime fields, one needs to work with a variant of \emph{Hasse derivatives} rather than with the more common notion of derivative.
}
\medskip
The proof of the second structural result (Theorem~\ref{thm:second structural result over Fq}) can be described informally as follows. Consider a degree $d$ polynomial $f$. Theorem~\ref{thm:structural result} implies the existence of an affine subspace $u_0 + U$ with dimension $\Omega(n^{1/(d-1)})$ on which $f$ is constant. One can then show (see Claim~\ref{claim:partition to lower degree}) that restricting $f$ to any affine shift of $U$ yields a degree (at most) $d-1$ polynomial. 
Thus, one can partition each such affine subspace recursively to obtain a partition of $\F_q^n$ to affine subspaces (not necessarily shifts of one another), such that $f$ is constant on each one of them.

In fact, to prove Theorem~\ref{thm:second structural result over Fq}, one is not required to find an affine subspace on which $f$ is constant, and it suffices to find an affine subspace on which the degree of $f$ decreases. In order to obtain the first algorithmic result (Theorem~\ref{thm:efficient structural result I}), we devise an algorithm that finds such an affine subspace and proceed similarly to the proof of Theorem~\ref{thm:second structural result over Fq}. To obtain the second algorithmic result (Theorem~\ref{thm:subexp alg}), we observe that the polynomial $t$ described above has many linear factors. This structure of $t$ allows us to save on the running time.

The generalization of Theorems~\ref{thm:structural result} and~\ref{thm:second structural result over Fq} to more than one polynomial is quite straightforward.

\section{Preliminaries}\label{sec:prelim}

We shall denote prime numbers with the letter $p$ and prime powers with $q$.
The set $\{1, \ldots, n\}$ is denoted by $[n]$. We denote by $\log(\cdot)$ the logarithm to the base $2$. Throughout the paper, for readability sake, we suppress flooring and ceiling. For $x,y \in \F_q^n$ we denote by $\langle x,y\rangle$ their scalar product over $\F_q$, i.e., $\langle x,y\rangle = \sum_{i=1}^n{x_i \cdot y_i}$. The vector $e_i$ is the unit vector defined as having $1$ in the $i^\text{th}$ entry and $0$ elsewhere. For a set $T\subseteq[n]$, we denote by $\one_T$ the indicating vector of $T$ with $1$ in the $i^\text{th}$ entry if $i \in T$ and $0$ otherwise.
For a vector $\alpha \in \N^m$, we denote its \emph{weight} by $\wt(\alpha) \triangleq \sum_{i}{\alpha_i}$.

The statistical distance between two random variables $X,Y$, over the same domain $D$, denoted by $\SD(X,Y)$, is defined as $\SD(X,Y) = \max_{A \subseteq D}{|\Pr[X \in A] - \Pr[Y \in A]|}$. It is known that $\SD(X,Y)$ is a metric. More precisely, it is (up to a multiplicative constant factor of $2$) the $\ell_1$ norm of the vector $(\Pr[d \in X] - \Pr[d \in Y])_{d \in D} \in \mathbb{R}^{|D|}$. In particular, we have the triangle inequality: for $X,Y,Z$ over $D$, $\SD(X,Z) \le \SD(X,Y) + \SD(Y,Z)$. Moreover, if $X$ can be written as a convex combination of two random variables $Y,Z$ as follows $X = (1-\gamma) \cdot Y + \gamma \cdot Z$, where $\gamma \in [0,1]$, then $\SD(X,Y) \le \gamma$. We sometimes abuse notation, and for a set $S \subseteq D$, consider $S$ also as the random variable that is uniformly distributed over the set $S$.

 \paragraph{Restriction to an affine subspace.}
 Let $f \colon \F_q^n \to \F_q$ be a function, $U \subseteq \F_q^n$ a subspace of dimension $k$ and $u_0 \in \F_q^n$ some vector. We denote by $f|_{u_0 + U} : (u_0 + U) \to \F_q$ the restriction of $f$ to $u_0 + U$. The degree of $f|_{u_0 + U}$ is defined as the minimal degree of a polynomial (from $\F_q^n$ to $\F_q$) that agrees with $f$ on $u_0 + U$.
For recursive arguments, it will be very useful to fix some basis $u_1, \ldots, u_k$ for $U$ and to consider the function $g: \F_q^k \to \F_q$ defined by
\[
g(x_1, \ldots, x_k) = f\(u_0 + \sum_{i=1}^{k}{x_i \cdot u_i}\).
\]
Note that the $\deg(g) = \deg(f|_{u_0 + U})$ regardless of the choice for the basis.

\paragraph{Polynomials.} We review some definitions and known facts about polynomials that we use.

The degree of a function $f:\F_q^n \to \F_q$, denoted by $\deg(f)$, is the degree of the unique multivariate
polynomial over $\F_q$, where each individual degree is at most $q-1$, which agrees with $f$ on $\F_q^n$. In the special case $q = 2$, such polynomials are called multi-linear. We will abuse notation and interchange between a function and its unique polynomial over $\F_q$ that agrees with $f$ on $\F_q^n$.

\begin{definition}
Let $f \colon \F_q^n \to \F_q$ be a polynomial of degree $d$, and let $\Delta \in \F_q^n$. The polynomial
\[
\frac{\partial f}{\partial{\Delta}} (x) \deff f(x+\Delta) - f(x),
\]
is called the \emph{derivative of $f$ in direction $\Delta$}.
\end{definition}
It is easy to verify that $\deg\(\frac{\partial{f}}{\partial{\Delta}}\) \le \deg(f) - 1$.
Let $\Delta_1, \ldots, \Delta_k \in \F_q^n$ then
\[
\frac{\partial^k{f}}{\partial\Delta_1 \ldots \partial\Delta_k}(x) = \sum_{S \subseteq[k]}{(-1)^{1+|S|} \cdot f\(x + \sum_{i\in S}{\Delta_i}\)}
\]
is a degree $\le \deg(f)-k$ polynomial.
The following lemma is a variant of the well-known the DeMillo-Lipton-Schwartz-Zippel lemma \cite{DL78,Schwartz80,Zippel79}.
\begin{lemma}[DeMillo-Lipton-Schwartz-Zippel]\label{lem:dlsz}
Let $q$ be a prime power. Let $f \in \F_q[x_1, \ldots, x_n]$ be a degree $d$ non-zero polynomial. Then,
\[
\Pr_{x \sim \F_q^n}[f(x_1, \ldots, x_n) \neq 0] \ge q^{-d/(q-1)}.
\]
\end{lemma}
\noindent
For completeness, we give the proof of Lemma~\ref{lem:dlsz} in Appendix~\ref{sec:dlsz}. The following folklore fact about polynomials over $\F_2$ is easy to verify.
\begin{fact}[M\"{o}bius inversion formula]\label{fact:mobius}
Let $f(x_1, \ldots, x_n) = \sum_{S \subseteq [n]} {a_S \cdot \prod_{i\in S}{x_i}}$ be a polynomial over $\F_2$. Then, its coefficients are given by the formula: $ a_S = \sum_{T \subseteq S} f(\one_T)$.
\end{fact}


\paragraph{Circuits.} A Boolean circuit is an unbounded fan-in circuit composed of $\OR$ and $\AND$ gates, and literals $x_i$, $\neg x_i$. The size of such a circuit is the number of gates in it. A Boolean formula is a Boolean circuit such that every $\OR$ and $\AND$ gate has fan-out $1$. De Morgan formula is a Boolean formula where each gate has fan-in at most $2$. We recall that an $\ACZERO$ circuit is a Boolean circuit of polynomial size and constant depth. An $\ACZEROPARITY$ circuit is an $\ACZERO$ circuit with unbounded fan-in $\XOR$ gates as well.



\section{Structural Results}

This section contains the proofs of all the structural results in this paper. In Section~\ref{sec:proof of first structural result} we give a proof for Theorem~\ref{thm:structural result}. Section~\ref{sec:proof of second structural result} contains the proof for Theorem~\ref{thm:second structural result over Fq}. The tightness of the first structural result is given in Section~\ref{sec:tightness}. In Section~\ref{sec:many poly} we describe the generalization of the two structural results to many polynomials.
In Section~\ref{sec:sparse} we prove Theorem~\ref{thm:sparse}. 

\subsection{Proof of Theorem~\ref{thm:structural result}}\label{sec:proof of first structural result}






In this section we prove Theorem~\ref{thm:structural result}. For a slightly simpler proof, for the special case $q=2$, we refer the reader to Appendix~\ref{sec:binary field}. The proof of Theorem~\ref{thm:structural result} is based on the following lemma.

\begin{lemma}\label{lemma:hasse}
Let $f \colon \F_q^n \to \F_q$ be some function, and let $U$ be a subspace of $\F_q^n$ with basis vectors $\Delta_1, \ldots, \Delta_k$. Then, there exist polynomials $(f_\alpha)_{\alpha \in \{0,1,\ldots, q-1\}^k}$ such that
\begin{enumerate}
\item $\deg(f_\alpha) \le \deg(f) - \wt(\alpha)$ for all $\alpha \in \{0,1,\ldots, q-1\}^k$.\label{item:1}
\item Let $x \in \F_q^n$, then $f|_{x + U}  \equiv 0$ if and only if $f_\alpha(x) = 0$ for all $\alpha \in \{0,1,\ldots, q-1\}^k$.\label{item:2}
\end{enumerate}
\end{lemma}

\begin{proof}

Complete $\Delta_1, \ldots, \Delta_k$ into a basis of $\F_q^n$ by picking vectors $\Delta_{k+1}, \ldots, \Delta_n \in \F_q^n$. Let $A$ be the linear transformation which maps the standard basis into $\Delta_1, \ldots ,\Delta_n$, and let $g(y) := f(Ay)$ (alternatively, $f(x) = g(A^{-1}x)$).
Write $g$ as a polynomial over $\F_q$:
\[
g(y) = \sum_{\gamma \in \{0,1,\ldots, q-1\}^n}{c_\gamma \cdot \prod_{i=1}^{n}{y_i^{\gamma_i}}}\;.
\]
Since both $f$ and $g$ can be obtained from one another by applying a linear transformation to the inputs, we have $\deg(f) = \deg(g)$.
Think of the input to $g$ as a concatenation of two parts $y = z \circ w$, where $z \in \F_q^k$, $w \in \F_q^{n-k}$.
Let $P_z: \F_q^n \to \F_q^{k}$ be the projection of a vector of length $n$ to the first $k$ coordinates and
let $P_w: \F_q^n \to \F_q^{n-k}$ be the projection to the last $n-k$ coordinates.
We may rewrite $g$ as
\[
g(z \circ w) = \sum_{\alpha \in \{0,1,\ldots, q-1\}^k}\sum_{\beta \in \{0,1,\ldots, q-1\}^{n-k}}{c_{\alpha \circ \beta} \cdot \prod_{i=1}^{k}{z_i^{\alpha_i}} \cdot \prod_{i=1}^{n-k}{w_i^{\beta_i}}}\;.
\]
By reordering the summations we get
\[
g(z \circ w) = \sum_{\alpha \in \{0,1,\ldots, q-1\}^k} {g_\alpha(w) \cdot \prod_{i=1}^{k}{z_i^{\alpha_i}}} \;,
\]
where \[g_\alpha(w) = \sum_{\beta \in \{0,1,\ldots, q-1\}^{n-k}}{c_{\alpha \circ \beta} \cdot   \prod_{i=1}^{n-k}{w_i^{\beta_i}}}\;.\]
Note that $\deg(g_\alpha) \le \deg(g) - \wt(\alpha)$.
We have
\begin{align*}
f|_{x+U} \equiv 0 \quad  &\iff \quad  g|_{A^{-1}x + A^{-1}U} \equiv 0\\
&\iff \quad g|_{A^{-1}x + \spa\{e_1, \ldots, e_k\}} \equiv 0 \;(*)\;.
\end{align*}
Writing $(z,w) = (P_z(A^{-1}x),P_w(A^{-1}x))$ gives
\begin{align*}
(*) \quad &\iff \quad \forall{z' \in \F_q^k}: g(z' \circ w) = 0\\
&\iff  \quad \forall \alpha: g_\alpha(w) = 0\\
&\iff  \quad \forall \alpha: g_\alpha(P_w(A^{-1}x)) = 0\;.
\end{align*}
Taking $f_\alpha$ to be the composition $g_{\alpha} \circ P_w \circ A^{-1}$ we obtain Item~\ref{item:2}. As $P_w \circ A^{-1}$ is simply a linear transformation, it is clear that
\[\deg(f_\alpha) \le \deg(g_\alpha) \le \deg(g) - \wt(\alpha) \le \deg(f) - \wt(\alpha)\;,
\] which completes the proof.
\end{proof}

\begin{proof}[Proof of Theorem \ref{thm:structural result}]

Assume without loss of generality that $f(u_0) = 0$, as otherwise we can look at the polynomial $g(x) = f(x)-f(u_0)$ which is of the same degree.
The proof is by induction.
Let $k$ be such that
\begin{equation}\label{eq:condition on k}
n > k + (d + 1) \cdot \sum_{j=0}^{d-1} {(d-j) \cdot \binom{k+j-1}{j}} \;.
\end{equation}
We assume by induction that there exists an affine subspace $u_0 + \spn\{\Delta_1, \ldots, \Delta_k\} \subseteq \F_q^n$, where the $\Delta_i$'s are linearly independent vectors, on which $f$ evaluates to 0.
Assuming Equation~\ref{eq:condition on k} holds, we show there exists a vector $\Delta_{k+1}$, linearly independent of $\Delta_1, \ldots, \Delta_k$, such that $f \equiv 0$ on $u_0 + \spn\{\Delta_1, \ldots, \Delta_{k+1}\}$. To this aim, consider the set
\[
A = \left\{
x\in \F_q^n\;\;\; \bigg\vert\;\;\; f|_{x + \spn\{\Delta_1, \ldots, \Delta_k\}} \equiv 0
\right\}.
\]
By the induction hypothesis, $u_0 \in A$. By Lemma~\ref{lemma:hasse}, for any $x \in \F_q^n$,
\[
f|_{x + \spn\{\Delta_1, \ldots, \Delta_k\}} \equiv 0 \quad \iff \quad
\forall{\alpha \in \{0,1,\ldots,q-1\}^k}: f_\alpha(x) = 0 \;,
\]
where $f_\alpha$ is of degree at most $d - \wt(\alpha)$.
\remove{\[
f_\alpha(x) \deff \frac{\partial{f}}{\partial^{\alpha_1}(\Delta_1) \ldots \partial^{\alpha_k}(\Delta_k)} (x).
\]
In particular, $\deg(f_\alpha) \le d-\wt(\alpha)$.}
Thus $f_\alpha \equiv 0$ for $\wt(\alpha) > d$, and we may write $A$ as
\[
A = \left\{
x \in \F_q^n \;\mid\; \forall{\alpha: \wt(\alpha)\le d} ,\; f_\alpha(x) = 0
\right\}.
\]
Hence, $A$ is the set of solutions to a system of $\le \binom{k+d}{d}$ polynomial equations,
where there are at most $\binom{k+j-1}{j}$ equations which correspond to $\alpha$'s of weight $j$ and thus to degree (at most) $d-j$ polynomials.
One can also write $A$ as the set of non-zeros to the single polynomial
\[
t(x) := \prod_{\alpha: \wt(\alpha) \le d}{\!\!\! (1-f_\alpha(x)^{q-1})} \,\,,
\]
which is of degree
\[
\deg(t) \le (q-1) \cdot \sum_{j=0}^{d-1} {(d-j) \cdot \binom{k+j-1}{j}}\;.
\]
Note that $t(x)$ obtains only the values $0$ and $1$. Let $R \subseteq \F_q$ be an arbitrary subset of $\F_q$ with size $|R| = \min(q, d+1)$.
Define a polynomial
\[
s(y) := \prod_{r \in R}{t(u_0 + r \cdot y)}
\]
We claim that any non-zero of $s$ not in the span of $\{\Delta_1, \ldots, \Delta_k\}$ can be taken to be the desired $\Delta_{k+1}$. Indeed, if $y$ is such that $s(y) = 1$, then $t(u_0 + r \cdot y) = 1$ for all $r \in R$. That is, for every $z \in \spn(\Delta_1, \ldots, \Delta_k)$ and any $r \in R$ it follows that $f(u_0 + z + r \cdot y) = 0$. Namely, $f$ obtains $|R|$ roots on the affine line with offset $u_0 + z$ and direction $y$. If $R = \F_q$ then clearly this implies that $f$ is the zero function restricted to the line. Otherwise, $|R| = d+1$ and thus $f$, which is a degree $d$ polynomial, obtains $d+1$ zeros on the line. Thus, again $f$ is the zero function on this line. Hence, $f(u_0 + z + r \cdot y) = 0$ for all $r \in \F_q$.



Thus, we just have to show that there exists some non-zero of $s$ which is linearly independent of $\{\Delta_1, \ldots, \Delta_{k}\}$.
Since the trivial solution $y = 0$ is a non-zero of $s$, we get that $s$ is not the constant $0$ function. Thus, by Lemma~\ref{lem:dlsz} it holds that
\[
\Pr[s(y) \neq 0] \ge q^{-\deg(s)/(q-1)}.
\]
The above equation implies that $s$ has at least $q^{n-\deg(s)/(q-1)}$ ones. Since we need to avoid $q^{k}$ linear combinations of the previous $\Delta_1, \ldots, \Delta_{k}$, it is enough to have
\begin{equation}\label{eq:k is small enough}
n - \frac{\deg(s)}{q-1} > k\;.
\end{equation}
Since
\[
\deg(s) \le (d+1) \cdot (q-1) \cdot \sum_{j=0}^{d-1} {(d-j) \cdot \binom{k+j-1}{j}}
\]
and by the assumption on $k$ in Equation~\eqref{eq:condition on k} we have that Equation~\eqref{eq:k is small enough} holds.
\end{proof}

\subsection{Proof of Theorem~\ref{thm:second structural result over Fq}}\label{sec:proof of second structural result}

In this section we prove Theorem~\ref{thm:second structural result over Fq}. To this end we use the following claim.

\begin{claim}\label{claim:partition to lower degree}
Let $q$ be a prime power. Let $f \colon \F_q^n \to \F_q$ be a degree $d$ polynomial. Assume there exists an affine subspace $u_0 + U$ of dimension $k$, restricted to which $f$ has degree at most $d-1$. Then, the degree of $f$ restricted to any affine shift of $U$ is at most $d-1$.
\end{claim}

\begin{proof}
Fix $u_1 \in \F_q^n$. Now, for any $u \in U$
$$
f(u_1 + u) = f(u_1 + u) - f(u_0 + u) + f(u_0+u) = \frac{\partial f}{\partial (u_1 - u_0)}(u_0 + u) + f(u_0+u).
$$
Since the degree of the partial derivative of $f$ is at most $d-1$ and the degree of $f|_{u_0 + U}$ is also at most $d-1$, we get that  $f|_{u_1+U}$ has degree at most $d-1$.
\end{proof}

\begin{proof}[Proof of Theorem~\ref{thm:second structural result over Fq}]
Let $\ConstStruct\in \(0,1\)$ be the constant from Theorem~\ref{thm:structural result}. Define the sequence $\{\beta_d \}_{d=1}^{\infty}$ as follows.
$$
\beta_d =
\left\{
  \begin{array}{ll}
    1/2, & \hbox{$d=1$;} \\
    \beta_{d-1} \cdot \ConstStruct^{\frac1{(d-2)!}}, & \hbox{$d > 1$.}
  \end{array}
\right.
$$
We will prove by induction on $d$, the degree of a given polynomial $f$, that there exists a partition of $\F_q^n$ to affine subspaces of dimension $\ge \beta_d \cdot n^{1/(d-1)!}$, such that $f$ restricted to each part is constant. The proof then follows by noting that for all $d \ge 1$,
$$
\beta_d = \frac1{2} \cdot \ConstStruct^{\frac1{(d-2)!} + \cdots + \frac1{1!} + \frac1{0!}} \ge \frac{\ConstStruct^{e}}{2},
$$
and thus one can take $\ConstStructTwo = \ConstStruct^{e} / 2$ to be the constant in the theorem statement.

The base case of the induction, namely $d=1$, trivially follows as $f$ is an affine function, and we can partition $\F_q^n$ to $q$ affine subspaces of dimension $n-1 \ge n/2 = \beta_1 n$, such that on each of which $f$ is constant. Assume now that $f$ is a degree $d > 1$ polynomial. By Theorem~\ref{thm:structural result} and Claim~\ref{claim:partition to lower degree}, there exists a partition of $\F_q^n$ to affine subspaces of dimension $k \ge \ConstStruct \cdot n^{1/(d-1)}$, such that $f$ restricted to any affine subspace in the partition has degree at most $d-1$.
Fix some affine subspace $u_0 + U$ in this partition, and apply the induction hypothesis to the polynomial $f' = f|_{u_0 + U}$, which has degree $d' \le d-1$.~\footnote{We may apply the induction because there exists a linear bijection from $U$ to $\F_q^{\dim{U}}$. More precisely, if $A$ is an $n \times k$ matrix over $\F_q$ that maps $U$ to $\F_q^k$ bijectively, then one can apply the induction to the polynomial $f''(x) = f'(u_0 + Ax)$, defined on $k$ variables, and then induce a partition of $u_0 + U$ from the partition of $\F_q^k$ obtained by the induction. The induction can be carried on $f''$ since $\deg{f''} \le \deg{f'} \le d-1$, where the first inequality holds because the variables of $f''$ are linear combinations of the variables of $f'$.}
By the induction hypothesis, we obtain a partition of $u_0 + U$ such that $f$ is constant on each part. Moreover, the dimension of each such part is at least
\[
\beta_{d'} \cdot k^{\frac1{(d'-1)!}} \ge
\beta_{d-1} \cdot k^{\frac1{(d-2)!}}
\ge \beta_{d-1} \cdot \left(\ConstStruct \cdot n^{\frac1{d-1}} \right)^{\frac1{(d-2)!}}
= \beta_{d-1} \cdot \ConstStruct^{\frac1{(d-2)!}} \cdot n^{\frac1{(d-1)!}}
= \beta_d \cdot n^{\frac1{(d-1)!}}\;,
\]
where the first inequality follows since $\{\beta_d\}_{d=1}^{\infty}$ is monotonically decreasing and $d'\le d-1$, and the last equality follows by the definitions of the $\beta_d$'s.
\end{proof}

\subsection{On the Tightness of Structural Result I}\label{sec:tightness}

Roughly speaking, Theorem~\ref{thm:structural result} states that for any prime power $q$, a degree $d$ polynomial over $\F_q$ in $n$ variables is not an affine disperser for dimension $k = \Omega(n^{1/(d-1)})$. We mentioned that this result is tight in the sense that by increasing $k$ a bit, there exists a degree $d$ polynomial which is an affine disperser. In this section we show, that in the special case $q = 2$, a stronger claim can be proven. Namely, by increasing $k$ a bit, there exists a degree $d$ polynomial which is an affine extractor.

\begin{theorem}\label{thm:tightness}
There exists a constant $c$ such that the following holds. Let $n,d$ be such that $d < n/2$. There exists a degree $d$ polynomial $f \colon \F_2^n \to \F_2$, such that for every affine subspace $u_0 + U \subseteq \F_2^n$ of dimension $k \ge c d \cdot n^{1/(d-1)}$, $\bias(f|_{u_0 + U}) \le 2^{-\Omega(k/d)}$.
\end{theorem}
\noindent
To prove Theorem \ref{thm:tightness} we apply the following lemma due to Ben-Eliezer, Hod and Lovett \cite{BHL09}.
\begin{lemma}[\cite{BHL09}, Lemma 2]\label{lemma:bhl}
Fix $\eps > 0$ and let $f \colon \F_2^n \to \F_2$ be a random degree $d$ polynomial~\footnote{That is, every monomial of degree at most $d$ appears in $f$ with probability $1/2$, independently of all other monomials.} for $d \le (1-\eps)n$. Then,
$$
\Pr_f\left[ \bias(f) > 2^{-c_1 n / d} \right] \le 2^{-c_2 \binom{n}{\le d}},
$$
where $0 < c_1, c_2 < 1$ are constants depending only on $\eps$.
\end{lemma}

\begin{proof}[Proof of Theorem \ref{thm:tightness}]
Let $f \colon \F_2^n \to \F_2$ be a random polynomial of degree at most $d$. Fix an affine subspace $u_0 + U \subseteq \F_2^n$ of dimension $k$. One can easily show that $f|_{u_0 + U}$ is equidistributed as a random polynomial on $k$ variables, of degree at most $d$. Therefore, by Lemma \ref{lemma:bhl},
$$
\Pr_f\left[ \bias(f|_{u_0 + U}) > 2^{-c_1 k / d} \right] \le 2^{-c_2\binom{k}{\le d}},
$$
where $c_1,c_2$ are the constants from Lemma \ref{lemma:bhl} suitable for the (somewhat arbitrary) choice $\eps = 1/2$. By taking the union bound over all $\le 2^n \cdot \binom{2^n}{k}$ affine subspaces of $\F_2^n$ of dimension $k$, it is enough to require that
$$
2^{-c_2 \binom{k}{\le d}} \cdot 2^n \cdot \binom{2^n}{k} < 1
$$
so to conclude the proof of the theorem. It is easy to verify that one can choose $c$, as a function of $c_2$, such that the above equation does hold for $k$ as defined in the theorem statement.
\end{proof}

\subsection{Generalization of the Structural Results to Many Polynomials}\label{sec:many poly}
\begin{theorem}[Structural Result I for many polynomials]\label{thm:structural result I for many polynomials}
Let $q$ be a prime power. Let $f_1, \ldots, f_t : \F_q^n \to \F_q$ be polynomials of degree $d_1, \ldots, d_t$ respectively. Let $k$ be the least integer satisfying the inequality
\[
n \le k + \sum_{i=1}^{t}{(d_i+1) \cdot \sum_{j=0}^{d_i-1}{(d_i-j) \cdot \binom{k + j - 1}{j} }}\;.
\]
Then, for every $u_0\in \F_q^n$ there exists a subspace $U \subseteq \F_q^n$ of dimension $k$, such that for all $i \in [t]$, $f_i$ restricted to $u_0 + U$ is a constant function. In particular, if $d_1, \ldots, d_t \le d$ then $k = \Omega((n/t)^{1/(d-1)})$. Moreover, for $d\le \log(n/t)/10$,  $k = \Omega(d \cdot (n/t)^{1/(d-1)})$.
\end{theorem}

Before proving Theorem~\ref{thm:structural result I for many polynomials} we note that by applying a probabilistic argument, it can be shown that the theorem is tight. In particular, it has the right dependency in the number of polynomials~$t$.

\begin{proof}
The proof is very similar to that of Theorem~\ref{thm:structural result}, so we only highlight the differences. As in the proof of~Theorem~\ref{thm:structural result}, we may assume that $f_1,\ldots,f_t$ evaluate to $0$ at $u_0$.
We build by induction an affine subspace $u_0 + U$ on which all the $t$ polynomials evaluate to $0$.
Given we already picked basis vectors $\Delta_1, \ldots, \Delta_{k}$, we consider the set $A$ to be the following:
\[
A = \left\{
x\in \F_q^n\;\;\; \bigg\vert\;\;\; \forall{i \in t},\,\,\, f_i|_{x + \spn\{\Delta_1, \ldots, \Delta_k\}} \equiv 0
\right\}.
\]
As in the proof of Theorem~\ref{thm:structural result}, $A$ can be written as the set of solutions
to a single polynomial equation $t(x) = 1$, where
\[
\deg(t) \le (q-1) \cdot \sum_{i=1}^{t} {(d_i+1) \sum_{j=0}^{d_i-1} {(d_i-j) \cdot \binom{k+j-1}{j}}},
\]
Similarly to Theorem~\ref{thm:structural result}, the polynomial $s$ is now defined, where $\deg(s) \le (d+1) \cdot \deg(t)$ and such that any non-zero of $s$, that is independent of $\Delta_1, \ldots, \Delta_k$, can be taken to be $\Delta_{k+1}$. By DeMillo-Lipton-Schwartz-Zippel lemma, it follows that as long $k$ is not too large, such a root can be found.
\end{proof}
\noindent
Similarly to the way we deduced Theorem~\ref{thm:second structural result over Fq} from Theorem~\ref{thm:structural result}, one can deduce the following theorem from Theorem~\ref{thm:structural result I for many polynomials}. We omit the proof.

\begin{theorem}[Structural Result II for many polynomials]\label{thm:structural result II for many polynomials}
Let $q$ be a prime power. Let $f_1, \ldots, f_t : \F_q^n \to \F_q$ be polynomials of degree at most $d$. Then, there exists a partition of $\F_q^n$ to affine subspaces, each of dimension $\Omega(n^{1/(d-1)!} / t^e)$, such that $f_1,\ldots,f_t$ are all constant on each part.
\end{theorem}

\subsection{Sparse Polynomials}\label{sec:sparse}

In this section we prove Theorem~\ref{thm:sparse}. To this end, we prove the following lemma.

\begin{lemma}\label{lem:reduction from sparse}
Let $f$ be a polynomial on $n$ variables over $\F_q$, with $n^c$ monomials. If $f$ is an affine disperser for dimension $k$, then there exists a subspace $U$ of dimension $\Omega(\sqrt{n})$ on which $f|_U$ is of degree at most $2(q-1)c$.
\end{lemma}

Lemma~\ref{lem:reduction from sparse} implies Theorem~\ref{thm:sparse}. Indeed, the above lemma states that for any polynomial $f$ on $n$ variables and $n^c$ monomials over $\F_q$, there exists an affine subspace of $\F_q^n$, with dimension $k(\Omega(\sqrt{n}), 2(q-1)c)$, on which $f$ is constant. By Theorem~\ref{thm:structural result}, $k(\Omega(\sqrt{n}), 2(q-1)c) = \Omega(n^{1/(4(q-1)c)})$, as desired.

\begin{proof}[Proof of Lemma~\ref{lem:reduction from sparse}]
We perform a random restriction to all variables $x_1, \ldots, x_n$. For each $i \in [n]$, independently, with probability $1-(2 \cdot n^c)^{-1/(2c)}$, we set $x_i$ to $0$. Consider a monomial that has at least $2c$ distinct variables. The probability that such a monomial survives the restriction is at most $1/(2 \cdot n^c)$. Thus, by the union bound, with probability at least $1/2$, no monomial with more than $2c$ distinct variables survived the restriction. Restricting ourselves to this event, since we may assume that the individual degree of each variable in the original polynomial is at most $q-1$, any surviving monomial has degree at most $2(q-1)c$.

The expected number of variables that survived the random restriction is $n \cdot (2 \cdot n^c)^{-1/(2c)} = \Omega(\sqrt{n})$. Thus, by the Chernoff bound, with probability at least, say, $3/4$, the number of surviving variables is $\Omega(\sqrt{n})$.

Thus, there exists a restriction of the variables that keeps $\Omega(\sqrt{n})$ of them alive, and such that the resulting polynomial has degree at most $2(q-1)c$. 
\end{proof}

\section{The Algorithmic Aspect}\label{sec:alg}

\subsection{Efficient Algorithm for Finding a Somewhat Large Subspace}

\begin{theorem}\label{thm:efficient structural result I}
Let $f \colon \F_2^n \to \F_2$ be a polynomial of degree $d\le \log(n)/3$ given as a black-box.
Then, there exists an algorithm that makes $\poly(n)$ queries to $f$, runs in time $\poly(n)$, and finds an affine subspace $U$ of dimension $\Omega(d \cdot n^{1/(d-1)})$ such that $\deg(f|_{U}) \le d-1$.
\end{theorem}

\noindent
The proof of Theorem~\ref{thm:efficient structural result I} is deferred to Appendix~\ref{sec:binary field algorithmic} as it relies on notations and ideas from the proof of the first structural result for the binary field, which can be found in Appendix~\ref{sec:binary field}. We advise the reader to look at the latter section before reading the proof of Theorem~\ref{thm:efficient structural result I}.

Theorem~\ref{thm:efficient structural result I} yields the following corollary.

\begin{corollary}
There exists an algorithm that given a degree $d$ polynomial $f: \F_2^n \to \F_2$ as a black box, runs in $\poly(n)$-time and finds an affine subspace of dimension $\Omega(n^{1/(d-1)!})$ on which $f$ is constant.
\end{corollary}

\remove{\begin{corollary}\label{cor:partition algorithm}
Let $f \colon \F_2^n \to \F_2$ be a degree $d$ polynomial given as a black-box, then there is a $2^{n-k} \cdot \poly(n)$-time $\poly(n)$-space algorithm, which partitions $\F_2^n$ to affine subspace of dimension $k$ on each of which $f$ is constant, where $k = \Omega(n^{1/(d-1)!})$.
\end{corollary}
\noindent
In particular, one can compute the number of satisfying assignments for $f$ using Corollary~\ref{cor:partition algorithm}.
}

\subsection{Subexponential-Time Algorithm for Finding an Optimal Subspace}

\begin{theorem}\label{thm:subexp alg}
There exists a constant $\beta>0$ such that the following holds.
There is an algorithm that given $f \colon \F_q^n \to \F_q$, a degree $d$ polynomial (as a list of monomials), where $3 \le d \le \log(n)/10$, and $u_0 \in \F_q^n$ as inputs, finds an affine subspace $u_0 + U$ of dimension $\Omega(k(n,d))$, restricted to which $f$ is constant. The algorithm runs in time $q^{\beta \cdot n^{(d-2)/(d-1)}} \cdot \poly(n^d)$, and uses $\poly(n^d,\log{q})$ space.
\end{theorem}
We obtain the following corollary.
\noindent
\begin{corollary}\label{cor:partition algorithm}
There exists a $q^{n-k} \cdot \poly(n^d)$-time $\poly(n^d, \log{q})$-space algorithm that given $f \colon \F_q^n \to \F_q$, a degree $d$ polynomial, partitions $\F_q^n$ to affine subspace of dimension $k$ on each of which $f$ is constant, where $k = \Omega(n^{1/(d-1)!})$.
\end{corollary}
\noindent
In particular, one can compute the number of satisfying assignments for $f$ using Corollary~\ref{cor:partition algorithm}.


\begin{proof}
We follow the proof of Theorem~\ref{thm:structural result}. Again, we may assume $f(u_0) = 0$. Given the previously chosen vectors $\Delta_1, \ldots, \Delta_k$ such that $f$ is the constant $0$ on $u_0 + \spa\{\Delta_1, \ldots, \Delta_k\}$, we show how to find a new vector $\Delta_{k+1}$ which is linearly independent of $\Delta_1, \ldots, \Delta_k$, such that $f$ is constantly zero on $u_0 + \spa\{\Delta_1, \ldots, \Delta_{k+1}\}$.
The set $A$ is the set of solutions to the following set of polynomial equations:
\[
\{f_\alpha(x) = 0 \;: \; \alpha \in \{0,1,\ldots, q-1\}^k, \wt(\alpha)\le d-1\}\;,
\]
and by our assumptions, $u_0$ is a solution to all of these equations.
By treating the polynomial $f$ as a formal sum of monomials we can calculate each $f_\alpha$ in $\poly(n^d)$ time.
Let $R$ be some arbitrary subset of $\F_q$ of size $\min\(q,d+1\)$ then any solution $y$ to the following set of equations which is linearly independent of $\Delta_1, \ldots, \Delta_k$ can be the new direction $\Delta_{k+1}$:
\[
\{f_\alpha(u_0 + r\cdot y) = 0 \;: \; \alpha \in \{0,1,\ldots, q-1\}^k, \wt(\alpha)\le d-1,\;r\in R\}\;.
\]
It is therefore enough to find more than $q^{k}$ different solutions to this set of equations, in order to guarantee that one of them will be linearly independent of the previous $\Delta_i$'s. In order to do so, we partition the set of equations into the set of linear equations and the set of non-linear equations:
\begin{align*}
L &= \{f_\alpha(u_0 + r\cdot y) = 0 \;: \; \alpha \in \{0,1,\ldots, q-1\}^k, \wt(\alpha)\le d-1, \deg(f_\alpha) = 1,\;r\in R\}\;.\\
NL &= \{f_\alpha(u_0 + r\cdot y) = 0 \;: \; \alpha \in \{0,1,\ldots, q-1\}^k, \wt(\alpha)\le d-1, \deg(f_\alpha) > 1,\;r\in R\}\;.
\end{align*}
Let $m = \sum_{f_\alpha \in NL} {\deg(f_\alpha)}$.
Since $0^n$ is a solution to all equations in $L \cup NL$, we can impose new linear equations which hold for $0^n$, keeping the system consistent. More specifically, we define a new set $L'$, which initially is equal to $L$, and iteratively add equations of the form $\{y_i = 0\}$ to $L'$ until $\dim(L') = n - m - k - 1$. \footnote{We add these constraints as concentrating at finding a solution of this form (that is, a solution that satisfies all equations in $L' \cup NL$ rather than only the equations in $L \cup NL$) is easier from the computational aspect.}

The set of solutions to both $L'$ and $NL$ is non-empty as it contains the all zeros vector. Furthermore, the sum of the degrees of equations in $L' \cup NL$ is exactly $(n - m - k - 1) + m = n - k - 1$. Therefore, by Lemma~\ref{lem:dlsz}, there are at least $q^{k+1}$ solutions to the equations in $L' \cup NL$, which guarantees that one of the solutions is linearly independent of $\Delta_1, \ldots, \Delta_k$.

Next, we show how to find all solutions to the equations in $L' \cup NL$. We find a basis for the set of solutions to $L'$ using Gaussian elimination, and iterate over all vectors in the affine subspace this basis spans. For each vector $y$ in this affine subspace we verify that all the equations in $NL$ are satisfied by $y$. The running time of this process is $O(q^{n-\dim(L')} \cdot |NL| \cdot n^d)$, which is $O(q^{m + k + 1} \cdot n \cdot n^{d})$.

As $m \le \min(d+1,q) \cdot \sum_{i=0}^{d-2} { (d-i) \cdot \binom{k+i-1}{i} }$, an elementary calculation shows that for $k \le  \frac{d}{10e} \cdot n^{1/(d-1)}$ and $3\le d\le \log(n)/10$ we have $m + k \le \beta \cdot n^{(d-2)/(d-1)}$ for some universal constant $\beta$.
Thus, the total running time of the algorithm is $q^{\beta \cdot n^{(d-2)/(d-1)}} \cdot \poly(n^d)$.
The algorithm uses $O((|NL| + |L|) \cdot n^{d} \cdot \polylog(q))$ space to store and manipulate the polynomials $f_\alpha$. In addition, $O(n^2 \cdot \polylog(q))$ space is used to perform the Gaussian elimination. Overall the space used by the algorithm is $O(n^{d+1} \cdot \polylog(q))$.
\end{proof}

\noindent

\section{Extractors and Dispersers for Varieties}

We start this section by proving Theorem~\ref{thm:from affine extractors to varieties for small t intro version}.

\begin{proof}[Proof of Theorem~\ref{thm:from affine extractors to varieties for small t intro version}]
Let $g_1, \ldots, g_t \colon \F_q^n \to \F_q$ be degree $d$ polynomials.
By Theorem~\ref{thm:structural result II for many polynomials}, there exists a partition of $\F_q^n$ to affine subspaces $P_1, \ldots, P_\ell$, each of dimension $\Omega(n^{1/(d-1)!}/t^e)$, such that $g_j|_{P_i}$ is constant for all $i \in [\ell]$ and $j \in [t]$. Since $f$ is an affine extractor for such dimension, with bias $\eps$, then for all $i \in [\ell]$ it holds that $\SD(f(P_i), \F_q) \le \eps$.

Let $I \subseteq [\ell]$ be the set of indices of affine subspaces in the partition such that $i \in I$ if and only if $g_j|_{P_i} = 0$ for all $j \in [t]$. In other words, we consider the partition of $\Variety(g_1,\ldots,g_t)$ to affine subspaces, induced by the partition of $\F_q^n$ to $P_1, \ldots, P_\ell$. Since the $P_i$'s are disjoint, the random variable $f(\Variety(g_1,\ldots,g_t)) = f(\cup_{i \in I}{P_i})$ is a convex combination of the random variables $\{f(P_i)\}_{i \in I}$. Thus, $\SD(f(\Variety(g_1,\ldots,g_t)), \F_q) \le \max_{i \in I}{\SD(f(P_i),\F_q)} \le \eps$.
\end{proof}
\noindent
We now give a formal statement and proof for the reduction from extractors for varieties to affine extractors, which does not depend on the number of polynomials defining the variety, but rather on the variety size.

\begin{theorem}\label{thm:from affine extractors to varieties for any t}
For every $d \in \mathbb{N}$ and $\delta, \rho \in (0,1)$ the following holds. Let $f \colon \F_q^n \to \F_q$ be an affine extractor for dimension $\Omega(n^{1/(d-1)!} / \ell^e)$ with bias $\eps$, where
$\ell = \log_{q}{(1/(\rho \delta))}$. Then, $f$ is an extractor with bias $\eps + \delta$ for varieties with density at least $\rho$ (i.e., size at least $\rho \cdot q^n$), that are the common zeros of any degree (at most) $d$ polynomials.
\end{theorem}


\begin{proof}
Let $g_1, \ldots, g_t \colon \F_q^n \to \F_q$ be degree (at most) $d$ polynomials. First, we prove the existence of $\ell$ polynomials $h_1, \ldots, h_{\ell} \colon \F_q^n \to \F_q$, each of degree at most $d$, with a variety that approximates $\Variety(g_1, \ldots, g_t)$. More precisely, we will have \begin{equation}\label{eq:requirement from g}
\Variety(g_1, \ldots, g_t) \subseteq \Variety(h_1, \ldots, h_{\ell}) \quad \text{and} \quad \Pr_{x \sim \F_q^n}[x \in \Variety(h_1, \ldots, h_{\ell}) \setminus \Variety(g_1, \ldots, g_t) ] \le q^{-\ell},
\end{equation}
The proof of this claim follows by a standard argument, like the one that appears in \cite{Razborov87, Smolensky87}: Let $\alpha_1, \ldots, \alpha_\ell$ be random vectors, sampled uniformly and independently from $\F_q^t$. For each $i \in [\ell]$, define the (random) polynomial
$$
H_i(x) = \sum_{j=1}^{t}{(\alpha_i)_j \cdot g_j(x)},
$$
where the summation and multiplications are taken over $\F_q$.
Clearly, if $x \in \Variety(g_1, \ldots, g_t)$ then $H_i(x) = 0$ with probability $1$ (where the probability is taken over $\alpha_1, \ldots, \alpha_\ell$). Otherwise, for each $i \in [\ell]$, $\Pr\left[ H_i(x) = 0 \right] = 1/q$. By an averaging argument, one can fix $\alpha_1, \ldots, \alpha_\ell$ and obtain fixed polynomials $h_1, \ldots, h_\ell$, of degree at most $d $, that satisfy the conditions in Equation \eqref{eq:requirement from g}.

Since $f$ is an affine extractor with bias $\eps$ for dimension $\Omega(n^{1/(d-1)!} / \ell^e)$, Theorem~\ref{thm:from affine extractors to varieties for small t intro version} implies that $\SD(f(\Variety(h_1, \ldots, h_\ell)),\F_q) \le \eps$. To conclude the proof, we show that
\[
\SD(f(\Variety(h_1, \ldots, h_\ell)), f(\Variety(g_1, \ldots, g_t))) \le \delta.
\]
To see this, observe that $\Variety(h_1, \ldots, h_\ell)$ can be written as a convex combination
\[
\Variety(h_1, \ldots, h_\ell) = \frac{|\Variety(g_1, \ldots, g_t)|}{|\Variety(h_1, \ldots, h_\ell)|} \cdot \Variety(g_1, \ldots, g_t) + \left( 1 - \frac{|\Variety(g_1, \ldots, g_t)|}{|\Variety(h_1, \ldots, h_\ell)|}\right) \cdot \mathcal{E},
\]
where $\mathcal{E}$ is some random variable over $\F_q$. Thus, by Equation~\eqref{eq:requirement from g},
\[
\SD(\Variety(h_1, \ldots, h_\ell), \Variety(g_1, \ldots, g_t)) \le 1 - \frac{|\Variety(g_1, \ldots, g_t)|}{|\Variety(h_1, \ldots, h_\ell)|} \le \frac{q^{-\ell}}{\rho} = \delta.
\]
This implies that
\[
\SD(f(\Variety(h_1, \ldots, h_\ell)), f(\Variety(g_1, \ldots, g_t))) \le \delta,
\]
as claimed.

\end{proof}

\noindent
Next, we prove Theorem~\ref{thm:from affine disperser to varieties for small t} which gives an analog reduction from dispersers for varieties to affine dispersers.


\begin{proof}[Proof of Theorem~\ref{thm:from affine disperser to varieties for small t}]
Let $g_1,\ldots,g_t \colon \F_q^n \to \F_q$ be degree (at most) $d$ polynomials. Let $u_0 \in \Variety(g_1,\ldots,g_t)$ (if $\Variety(g_1,\ldots,g_t) = \emptyset$, there is nothing to prove). By Theorem~\ref{thm:structural result I for many polynomials}, there exists a subspace $U$ of dimension $\Omega(d\cdot(n/t)^{1/(d-1)})$ such that $u_0 + U \subseteq \Variety(g_1,\ldots,g_t)$. The proof then follows as $f$ is an affine disperser for dimension $\Omega(d\cdot(n/t)^{1/(d-1)})$.
\end{proof}

\remove{\noindent
In~\cite{S11} (Theorem 1.2), an explicit construction of an affine disperser $f \colon \F_2^n \to \F_2$ for dimension $2^{\log^{0.9}{n}}$ is given. Theorem~\ref{thm:explicit disperser intro version} readily follows by Theorem~\ref{thm:from affine disperser to varieties for small t} as indeed, one only needs to make sure that $(n/t)^{1/(d-1)} = \Omega(2^{\log^{0.9}{n}})$.
}

\section{From Affine Dispersers to Affine Extractors}\label{sec:reduction from extractors to dispersers}

To prove Theorem~\ref{thm:reduction from extractors to dispersers intro version}, we use the following theorem of Kaufman and Lovett \cite{KL08}.

\begin{theorem}[\cite{KL08}]\label{thm:KL}
Let $p$ be a prime number and let $f \colon \F_p^n \to \F_p$ be a degree (at most) $d$ polynomial with $\bias(f) \ge \delta$. Then, there exist $c = c(d, \delta)$
polynomials $f_1, \ldots, f_c$ of degree at most $d-1$ such that $f = G(f_1, \ldots, f_c)$, for some function $G: \F_p^c \to \F_p$. Moreover, $f_1, \ldots, f_c$ are derivatives of the form $\frac{\partial f }{\partial y}$ where  $y \in \F_p^n$.
\end{theorem}
\noindent
\begin{proof}[Proof of Theorem~\ref{thm:reduction from extractors to dispersers intro version}]
We show by a counter-positive argument that if $f$ is not an affine extractor for dimension $k'$ with bias $\delta$, then $f$ is not an affine disperser for dimension $k$.
Let $f: \F_p^n \to \F_p$ be a function which is not an affine extractor for dimension $k'$ with bias $\delta$.
Then, there exists an affine subspace $u_0 + U$, with $\dim(U) = k'$ such that $\bias(f|_{u_0+U})>\delta$.
Let $u_1, \ldots, u_{k'}$ be a basis for $U$ and let $g : \F_p^{k'} \to \F_p$ be the function defined by $g(y_1, \ldots, y_{k'}) = f(u_0 + \sum_{i=1}^{k'}{u_i \cdot y_i})$. Then, $g$ is a $\delta$-biased polynomial of degree $\le d$.
Applying Theorem~\ref{thm:KL} to $g$, we can write it as $G(g_1, \ldots, g_c)$, where the $g_i$'s are of degree at most $d-1$, and $c = c(d, \delta)$ as defined in Theorem~\ref{thm:KL}.

By Theorem~\ref{thm:structural result I for many polynomials}, there is an affine subspace $W$ of $\F_p^{k'}$ with dimension $\ConstStruct \cdot (k'/c)^{1/(d-2)}$ for which all the $g_i$'s are constant, for some constant $\ConstStruct > 0$. In particular $g|_W$ is constant, which implies that there exists a subspace of  $\F_p^n$, with the same dimension, on which the original function $f$ is constant. Taking $k' = k^{d-2} \cdot \frac{c(d, \delta)}{\ConstStruct^{d-2}}$ completes the proof.
\end{proof}
\noindent
For degree $3$ and $4$, we rely on stronger results from~\cite{HS10}.
Although degree $3$ was treated in~\cite{BSK12}, we present it here for completeness.
\begin{theorem}\label{thm:degree 3 4}
Let $f \colon \F_p^n \to \F_p$ be an affine disperser for dimension $k$ of degree $d$.
If $d=3$ then $f$ is an affine extractor for dimension $k' = k+O(\log(1/\delta)^2)$ with bias $\delta$.
If $d=4$ then $f$ is an affine extractor for dimension $k' = k\cdot \poly(1/\delta)$ with bias $\delta$.
\end{theorem}

\begin{proof}

As in the proof of Theorem~\ref{thm:reduction from extractors to dispersers intro version}, it is enough to show that if $g$ is a degree $3$ or $4$ polynomial over $\F_p$ with $k'$ variables and bias $\ge \delta$ then there exists a subspace of dimension $k$ on which $g$ is constant.
We consider the two cases $\deg(f)=3,4$ separately.

\paragraph{Cubic ($\deg(g)=3$).}
Implicit in \cite{HS10}, any polynomial of degree $3$ with bias $\ge\delta$, in particular $g$, can be represented as
\[
\sum_{i=1}^{r}{\ell_i(x) \cdot q_i(x)} + q_0(x),\]
where the $\ell_i$'s are linearly independent linear functions (with no constant term), $\deg(q_i) \le 2$ and $r = O(\log^2(1/\delta))$.
Restricting to the subspace $W$ defined by $\{x: \ell_i(x)=0\}$ reduces the degree of $g$ to at most $2$, and by Claim~\ref{claim:partition to lower degree}, this is also true for any coset of this subspace. By averaging, there is a coset on which $\bias(g|_{w+W}) \ge \delta$. By Dickson's theorem \cite{D01}, there is an affine subspace $w' + W'$ of $w+W$ of co-dimension $O(\log(1/\delta))$ on which $g$ is constant. Setting $k' = k + O(\log^2(1/\delta))$ ensures that $\dim(W')$ is at least $k$.

\paragraph{Quartic, ($\deg(g)=4$).}
Theorem~4 in \cite{HS10} states that any polynomial of degree $4$ with bias $\ge\delta$, in particular $g$, can be represented as
\[
\sum_{i=1}^{r}{\ell_i(x) \cdot g_i(x)} + \sum_{i=1}^{r}{q_i(x)\cdot q'_i(x)} + g_0(x),
\]
where $\deg(\ell_i)\le 1, \deg(q_i)\le 2, \deg(q'_i) \le 2, \deg(g_i) \le 3$ and $r = \poly(1/\delta)$.
By Theorem~\ref{thm:structural result I for many polynomials}, there exists a subspace $W$ of dimension $\Omega(n/r)$ on which all $\ell_i$'s, $q_i$'s and $q'_i$'s are constants. By Claim~\ref{claim:partition to lower degree}, in any coset of $W$ the degrees of $\ell_i$, $q_i$ and $q'_i$ for $i=1,\ldots,r$ are decreased by at least 1, hence $g|_{w+W}$ is of degree at most $3$ for any coset $w+W$.
Since $\bias(g) \ge \delta$, by averaging there is a coset on which $\bias(g|_{w+W}) \ge \delta$. Using the earlier case of biased cubic polynomials, there is an affine subspace $w' + W'$ of dimension $\Omega(n/r) - O(\log^2(1/\delta))$ on which $g$ is constant.
Setting $k' = k\cdot \poly(1/\delta)$ ensures that the dimension of $W'$ is at least $k$.
\end{proof}

\paragraph{Remark:}
It may be tempting to think that the polynomial loss of parameters in our reduction from affine extractors to affine dispersers, $k' = O_{\delta,d}(k^{d-2})$,  is not necessary. Indeed, Theorem~\ref{thm:degree 3 4} shows that for degree $3$ and $4$ one can take the dimension $k'$ of the affine extractor (for a constant error, say) to be linear in $k$ -- the dimension of the affine disperser. However, this linear dependency breaks for $d \ge 6$, as pointed up to us by Shachar Lovett. To see this, take $f: \F_2^n \to \F_2$ to be the product of two random degree $3$ polynomials. It is easy to check that, with high probability, $f$ is an affine disperser for dimension $\Theta(\sqrt{n})$, whereas $\Pr[f = 1] = 1/4 + o(1)$. Namely, $f$ is not even an $(n,n)$ affine extractor.

Nonetheless, a better polynomial dependency may still be possible. Perhaps $k' = O_{\delta,d}(k^{(d-2)/2})$ (which is not ruled out by similar counterexamples). 
\section{$\ACZEROPARITY$ Circuits and Affine Extractors / Dispersers}

In Section~\ref{sec:aczero and disperses} we (easily) derive lower bounds on the dimension for which an $\ACZERO$ circuit can be affine disperser.
In Section~\ref{sec:depth2 cannot be affine dispersers} we prove that a depth $2$ $\ACZEROPARITY$ circuit on $n$ inputs cannot compute an affine disperser for dimension $n^{o(1)}$. We do so by a reduction to Theorem \ref{thm:structural result}.
%

\subsection{$\ACZERO$ Circuits Cannot Compute Affine Dispersers for Dimension $o(n/\polylog(n))$}\label{sec:aczero and disperses}
The next lemma, following H{\aa}stad's work~\cite{Hastad86}, appears in \cite{BoppanaS90}.
\begin{lemma}[\cite{BoppanaS90}, Corollary 3.7, restated]\label{lemma:BP}
Let $f \colon \F_2^n \to \F_2$ be a function computable by a depth $d$ and size $s$ Boolean circuit. Then, there is a restriction $\rho$ leaving $\frac{n}{10(10\log(s))^{d-2}} \;-\; \log(s)$ variables alive, under which $f|_{\rho}$ is constant.
\end{lemma}
\noindent
Lemma~\ref{lemma:BP} readily implies the following corollary.
\begin{corollary}\label{cor:ACZERO}
Let $f \colon \F_2^n \to \F_2$ be a function computable by a Boolean circuit of depth $d$ and size $s$. Then, $f$ cannot be a bit fixing disperser (and, in particular, $f$ cannot be an affine disperser) for min-entropy $k < \frac{n}{10(10\log(s))^{d-2}} - \log(s)$.
\end{corollary}

\subsection{Depth 2 $\ACZEROPARITY$ Circuits Cannot Compute Good Affine Dispersers}\label{sec:depth2 cannot be affine dispersers}

As mentioned in the introduction, to prove Theorem~\ref{thm:depth 2 are not dispersers intro version}, one only needs to prove Lemma~\ref{lemma:reduction from depth 2 to polynomials intro version}.

\begin{proof}[Proof of Lemma~\ref{lemma:reduction from depth 2 to polynomials intro version}]
During the proof we will exploit the fact that if a function $f$ on $n$ inputs is an affine disperser for dimension $k$, then fixing the values of $m$ inputs or even the values of $m$ linear functions on the inputs, one gets an affine disperser on $n-m$ inputs for the same dimension $k$.

We assume that the top gate is an $\XOR$ gate. Afterwards we justify this assumption by showing that if the top gate is not an $\XOR$ gate, then the circuit $C$ could not have computed an affine disperser with the claimed parameters to begin with.

Note that one might as well assume that there are no $\XOR$ gates at the bottom level. Indeed, assume there are $t$ $\XOR$ gates at the bottom level, and denote by $\ell_1,\ldots,\ell_t$ the linear functions computed by these gates, respectively. Define the linear function $\ell = \ell_1 \oplus \cdots \oplus \ell_t$. Note that if $\ell$ is the constant $1$ then by removing all the $t$ gates from $C$ and wiring the constant $1$ as an input to the top gate, one gets an equivalent circuit with no $\XOR$ gates at the bottom layer. Assume therefore that $\ell$ is not the constant $1$. Then, by removing all the $\XOR$ gates at the bottom layer, we get a circuit, with no $\XOR$ gates at the bottom layer, that is equivalent to the original circuit on the affine subspace $\{ x : \ell(x) = 0 \}$. Hence, the resulting circuit is an affine disperser on $n-1$ inputs for dimension $k$.

We perform a random restriction to all variables, leaving a variable alive with probability $p = \frac{1}{4\sqrt{n}}$ and otherwise setting the value of a variable uniformly and independently at random.
We show that the restriction shrinks all $\OR,\AND$ gates to have fan-in smaller than $2c$ with positive probability. We consider $\AND$ gates, but our arguments may be carried to $\OR$ gates similarly.
The restriction shrinks every $\AND$ gate in the following way: if one of the literals which is an input to the $\AND$ gate is false under the restriction, the $\AND$ gate is eliminated. Otherwise, the $\AND$ gate shrinks to be the $\AND$ of all the remaining live variables. We wish to bound the probability that each $\AND$ gate is of fan-in greater than $2c$ after the restriction. Let $m$ be the fan-in of the $\AND$ gate before the restriction, and $m'$ its fan-in afterwards.
We have
\[
\Pr[ m' \ge 2c ] = \sum_{i=2c}^{m} { \binom{m}{i} \cdot p^{i} \cdot \(\frac{1-p}{2}\)^{m-i}}
\le \sum_{i=2c}^{m} { \binom{m}{i} \cdot p^{i} \cdot (1/2)^{m-i}}
= (1/2)^m \cdot \sum_{i=2c}^{m} { \binom{m}{i} \cdot (2p)^{i}} \;.
\]
Since $2p$ is smaller than $1$, the right hand side of the above inequality is at most $(1/2)^m \cdot 2^{m} \cdot (2p)^{2c} = (2p)^{2c}$. Thus, $\Pr[ m' \ge 2c] \le (2p)^{2c}$.
By our choice of parameter $p$, this is at most $1/(4n)^c$. By union bound over all $\le n^c$ $\AND$ and $\OR$ gates, with probability at least $1-1/4^c \ge 3/4$ over the random restrictions, the fan-in of all $\AND$ and $\OR$ gates, under the restriction, is smaller than $2c$.
Furthermore, by Chernoff bound, with probability greater than $1/2$ over the random restrictions, the number of surviving variables is at least $\sqrt{n}/5$. Therefore, there exists a restriction where the number of surviving variables is $\sqrt{n}/5$ and all $\AND$ and $\OR$ gates in the resulting circuit, under the restriction, have fan-in smaller than $2c$.
Expressing the resulting circuit as a polynomial over $\F_2$ we get a polynomial on at least $\sqrt{n}/5$ variables with degree at most $2c$ which is an affine disperser for dimension $k$.

We are left to justify the assumption that the top gate must be an $\XOR$ gate. For contradiction, assume that the top gate is an $\OR$ gate. The case where the top gate is an $\AND$ gate is handled similarly. If there is an $\XOR$ gate at the bottom layer of $C$, we choose such gate and consider the affine subspace of co-dimension $1$ on which this $\XOR$ gate outputs $1$. Since the top gate is an $\OR$ gate, the circuit $C$ is the constant $1$ on an affine subspace of co-dimension $1$. This stands in contradiction as $k$ is (much) smaller than $n-1$. Thus, we obtain a depth 2 $\ACZERO$ circuit with size $s = n^c$. However, under the assumption that $k < n/10-\log(s)$ this is a contradiction to Corollary~\ref{cor:ACZERO}.

\end{proof}

\section*{Acknowledgement}
We wish to thank our advisor Ran Raz for many helpful discussions and for his encouragement. We thank Chaim Even Zohar, Elad Haramaty, Noam Lifshitz and Amir Shpilka for helpful discussions regarding this work. We thank the user goes by the name david from stack exchange for pointing out \cite{BHL09}. We thank the anonymous referees for pointing out \cite{TB98} and for many helpful comments.

\bibliographystyle{alpha}
\bibliography{bibliography}

\appendix
\section{Depth 3 $\ACZEROPARITY$ Circuits Can Compute Optimal Affine Extractors}\label{sec:depth3 can be affine dispersers}
We start this section by giving a proof for the following folklore claim. We bother doing so because afterwards we argue that the proof implies, in fact, something stronger, which we make use of.
\begin{claim}\label{claim:affine extractors}
There exist universal constants $n_0, c$ such that the following holds. For every $\eps > 0$ and $n > n_0$ there exists an affine extractor for dimension $k$ with bias $\eps$, $f : \F_2^n \to \F_2$, where $k = \log{\frac{n}{\eps^2}} + \log{\log{\frac{n}{\eps^2}}} + c$.
\end{claim}
\noindent
The proof of Claim \ref{claim:affine extractors} makes use of Hoeffding bound.

\begin{theorem}[Hoeffding Bound]\label{thm:hoeffding}
Let $X_1, \ldots, X_n$ be independent random variables for which $X_i \in [a_i,b_i]$. Define $X = \frac1{n} \cdot \sum_{i=1}^{n}{X_i}$, and let $\mu = \Exp[X]$. Then,
$$
\Pr[|X-\mu| \ge \eps] \le 2 \cdot \exp\left(- \frac{2n^2\eps^2}{\sum_{i=1}^{n}{(b_i-a_i)^2}}\right).
$$
\end{theorem}

\begin{proof}[Proof of Claim \ref{claim:affine extractors}]
Let $F : \F_2^n \to \F_2$ be a random function, that is, $\{F(x)\}_{x \in \F_2^n}$ are independent
random bits. Fix an affine subspace $u_0 + U \subseteq \ftwo^n$ of dimension $k$ as defined above. By Hoeffding Bound (Theorem \ref{thm:hoeffding}),
$$
\Pr\left[\frac1{2^k} \left|\sum_{u \in u_0 + U}{(-1)^{F(u)}}\right| \ge \eps \right] \le 2 \cdot \exp\left(- \frac{2^k \eps^2}{2}\right).
$$
The number of affine subspaces of dimension $k$ is bounded by $2^n \binom{2^n}{k} \le 2^{(k+1)n}$. Hence, by union bound over all affine subspaces, if $2^{(k+1)n} \cdot 2 e^{-2^k \eps^2 / 2} <1$ then there exists a function $f : \F_2^n \to \F_2$ that is an affine extractor for dimension $k$ with bias $\eps$. It is a simple calculation to show that our choice of $k$ suffices for the above equation to hold.
\end{proof}
\noindent
For the proof of Theorem~\ref{thm:affine construction}, we introduce the following notion.

\begin{definition}
An $(n,k,d)$ \emph{linear injector} with size $m$ is a family of $d \times n$ matrices $\{A_1, \ldots , A_m \}$ over $\ftwo$ with the following property: for every subspace $U \subseteq \ftwo^n$ of dimension $k$, there exists an $i \in [m]$ such that $\ker(A_i) \cap U = \{0\}$.
\end{definition}

\begin{lemma}\label{lemma:linear injectors}
For every $n,k$ such that $2 \le k \le n$, there exists an $(n,k,k+1)$ linear injector with size $m = nk$.
\end{lemma}

\begin{proof}
Fix a subspace $U \subseteq \ftwo^n$ of dimension $k$. Let $A$ be a $d \times n$ matrix such that every entry of $A$ is sampled from $\ftwo$ uniformly and independently at random. For every $u \in U \setminus \{0\}$ it holds that $\Pr[Au = 0] = 2^{-d}$. By taking the union bound over all elements in $U\setminus\{0\}$, we get that
\[
\Pr[\ker(A) \cap U \neq \{0\}] \le 2^{k-d}.
\]
Let $A_1, \ldots , A_m$ be $d \times n$ matrices such that the entry of each of the matrices is sampled from $\ftwo$ uniformly and independently at random. By the above equation, it holds that
\[
\Pr[\forall i \in [m] \,\, \ker(A_i) \cap U \neq \{0\}] \le 2^{m(k-d)}.
\]
The number of linear subspaces of dimension $k$ is bounded above by $\binom{2^n}{k}$, which is bounded above by $2^{nk-1}$ for $k \ge 2$. Thus, if $2^{nk-1} \cdot 2^{m(k-d)} < 1$
there exists an $(n,k,d)$ linear injector with size $m$. The latter equation holds for $d=k+1$ and $m=nk$.
\end{proof}

\begin{lemma}\label{lemma:the affine subspace independent random function}
Let $n_0, c$ be the constants from Claim \ref{claim:affine extractors}. Let $n > n_0$ and let $k,\eps$ be such that $k = \log{\frac{n}{\eps^2}} + \log{\log{\frac{n}{\eps^2}}} + c$. Let $\{A_1, \ldots, A_m\}$ be an $(n,k,d)$ linear injector with size $m$. Then, there exist functions $f_1, \ldots, f_m : \F_2^d \to \F_2$ such that the function $f : \F_2^n \to \F_2$ defined by
\begin{equation}\label{eq:function of affine extractor}
f(x) = \bigoplus_{i=1}^{m}{f_i(A_i x)}
\end{equation}
is an affine extractor for dimension $k$ with bias $\eps$.
\end{lemma}

\begin{proof}
Recall that in the proof of Claim \ref{claim:affine extractors}, we took $F$ to be a random function. We observe however, that the proof did not use the full independence offered by a uniformly sampled random function. In fact, the proof required only that for every affine subspace $u_0 + U \subseteq \ftwo^n$ of dimension $k$, $\{f(u) \}_{u \in u_0 + U}$ are independent random bits.

Let $F_1, \ldots, F_m : \F_2^d \to \F_2$ be independent random functions, that is, the random bits $\{ F_i(x)\}_{i \in [m], x \in \F_2^d}$ are independent. Define the random function $F : \F_2^n \to \F_2$ as follows
$$
F(x) = \bigoplus_{i=1}^{m}{F_i(A_i x)}.
$$
We claim that for every affine subspace $u_0 + U \subseteq \ftwo^n$ of dimension $k$, the random bits $\{ F(u) \}_{u \in u_0 + U}$ are independent. By the observation above, proving this will conclude the proof. Let $u_0 + U\subseteq \ftwo^n$ be an affine subspace of dimension $k$. As $\{A_1, \ldots, A_m \}$ is an $(n,k,d)$ linear injector, there exists an $i \in [m]$ such that $\ker(A_i) \cap U = \{0\}$. This implies that for every two distinct elements $u,v \in U$ it holds that $A_i (u_0 + u) \neq A_i (u_0 + v)$. Otherwise $A_i(u+v) = 0$ and thus $u+v$, a non-zero vector in $U$, lies in $\ker(A_i)$. This stands in contradiction to the choice of $i$. Recall that $F_i$ is a random function, and from the above it follows that $A_i$ behaves as an injection to the domain $u_0 + U$. Hence, the random bits $\{F_i(A_i u)\}_{u \in u_0 + U}$ are independent. Since $F(x)$ is defined to be the $\XOR$ of $F_i(A_i x)$ with $m-1$ other \emph{independent} random variables, we get that $\{F(u)\}_{u \in u_0 + U}$ are also independent random bits, as claimed.
\end{proof}

\begin{theorem}\label{thm:affine construction}
Let $f$ be the function from Equation~\eqref{eq:function of affine extractor}, where $\{ A_1, \ldots, A_m \}$ is the $(n,k,d)$ linear injector from Lemma~\ref{lemma:linear injectors} (that is, $m=nk$ and $d=k+1$). Then, $f$ is an affine extractor for dimension $k$ and bias $\eps$, where $k = \log{(n/\eps^2)} + \log{\log{(n/\eps^2)}} + O(1)$. Moreover,
\begin{enumerate}
\item $\deg(f) = \log{(n/\eps^2)} + \log{\log{(n/\eps^2)}} + O(1)$.
\item $f$ can be realized by an $\XAX$ circuit of size $O((n/\eps)^2 \cdot \log^3{(n/\eps)})$.
\item $f$ can be realized by a De Morgan formula of size $O((n^5 / \eps^2) \cdot \log^3{(n/\eps)})$.
\end{enumerate}
\end{theorem}
\begin{proof}
To prove the first item, we note that each of the $f_i$'s is a function on $d=k+1$ inputs, and thus can be computed by a polynomial with degree at most $k+1$. The proof then follows as in the computation of $f$, each $f_i$ is composed with linear functions of the variables, and $f$ is the $\XOR$ of the $f_i$'s.
%


To prove the second item, we show an $\XAX$ circuit $C$ with the desired size, that computes the function $f$.
Since each of the functions $f_i$ are degree $d$ polynomials on $d$ inputs, each of them can be computed by an $\XOR-\AND$ circuit, where the fan-in of the top $\XOR$ gate is bounded above by $2^d$ and the fan-in of each $\AND$ gate is at most $d$.
Thus, for $i \in [m]$, each of the functions $f_i(A_i x)$ on $n$ inputs is computable by an $\XOR-\AND-\XOR$ circuit.

By its definition, $f$ is the $\XOR$ of these functions and so one can collapse this $\XOR$ together with the top $m$ $\XOR$ gates. This yields an $\XOR-\AND-\XOR$ circuit $C$ that computes $f$.

The size of the circuit $C$ is $O(m \cdot d \cdot 2^d)$ as each of the $m$ functions $f_i(A_i x)$ applies $2^d$ $\AND$  gates, each on $d$ $\XOR$ gates (whom in turn compute the linear injector). Since $m=nk$ and $d=k+1$, $\size(C) = O((n/\eps)^2 \cdot \log^3(n/\eps))$ as stated.
%
%
%
%


As for the third item, we show a De Morgan formula with the desired size, that computes $f$. Since each of the functions $f_i$ are on $d$ inputs, each of them can be computed by a De Morgan formula of size $O(2^d)$. Moreover, every $\XOR$ operation needed for the computation of the linear injector $\{A_1, \ldots, A_m\}$ can be implemented in size $O(n^2)$.
Replacing each leaf in the formula for $f_i$ with the relevant formula computing the corresponding bit of $A_i  x$ (or its negation), results in an $O(2^d n^2)$ size De Morgan formula computing $f_i(A_i x)$.
Again, since the $\XOR$ of bits $y_1, \ldots,y_m$ can be computed by a De Morgan formula of size $O(m^2)$, and one can replace each leaf marked by $y_i$ (or $\neg y_i$) with the formula computing $f_i(A_i x)$ (or its negation), one gets a De Morgan formula computing $f$ of size
\[
 O(m^2 \cdot 2^d \cdot n^2)
 = O((nk)^2 \cdot 2^k \cdot n^2)
 = O((n^5 / \eps^2) \cdot \log^3(n/\eps)),
\]
as desired.
\end{proof}
\section{A Slightly Simpler Proof of the First Structural Result for $\F_2$}\label{sec:binary field}

In this section we give a slightly simpler proof for Theorem~\ref{thm:structural result}, for the special case $q = 2$. We prove the following:

\begin{theorem}[Structural Result I for the Binary Field]\label{thm:structural result binary field}
Let $k$ be the smallest integer such that
\[
n \le k + \sum_{j=0}^{d-1} {(d-j) \cdot \binom{k}{j}}\;.
\]
Let $f \colon \F_2^n \to \F_2$ be a degree $d$ polynomial, and let $u_0 \in \F_2^n$. Then, there exists a subspace $U \subset \F_2^n$ of dimension $k$ such that $f|_{u_0+U}$ is constant.
\end{theorem}
\noindent

\begin{proof}
Fix $u_0\in \F_2^n$.
We assume without loss of generality that $f(u_0) = 0$, as otherwise we can look at the polynomial $g(x) = f(x)-f(u_0)$ which is of the same degree.
The proof is by induction.
Let $k$ be such that
\begin{equation}\label{eq:condition on k}
n > k + \sum_{j=0}^{d-1} {(d-j) \cdot \binom{k}{j}} \;.
\end{equation}
We assume by induction that there exists an affine subspace $u_0 + \spa\{\Delta_1, \ldots, \Delta_k\} \subseteq \F_2^n$, where the $\Delta_i$'s are linearly independent vectors on which $f$ evaluates to 0.
Assuming Equation~\ref{eq:condition on k} holds, we show there exists a vector $\Delta_{k+1}$, linearly independent of $\Delta_1, \ldots, \Delta_k$, such that $f \equiv 0$ on $u_0 + \spa\{\Delta_1, \ldots, \Delta_{k+1}\}$. To this aim, consider the set
\[
A = \left\{
x\in \F_2^n\;\;\; \bigg\vert\;\;\; \forall{S\subseteq[k]},\;{f\(x + \sum_{i \in S} {\Delta_i}\) = 0}
\right\}.
\]
By the induction hypothesis, $u_0 \in A$. It can be verified that for any $x \in \F_2^n$
\[
\forall{S\subseteq[k]}: f\(x + \sum_{i \in S} {\Delta_i}\) = 0 \quad \Leftrightarrow \quad
\forall{S\subseteq[k]}: f_S(x) = 0 \;,
\]
where $f_S$ is defined by
\[
f_S(x) \deff \sum_{T\subseteq S} {f\(x + \sum_{i \in T} {\Delta_i}\)}.
\]
Namely, $f_S$ is the derivative of $f$ in directions $\{\Delta_i \}_{i \in S}$. In particular, $\deg(f_S) \le d-|S|$. Thus $f_S \equiv 0$ for $|S| > d$, and we may write $A$ as
\[
A = \left\{
x\in \F_2^n \;\mid\; \forall{S\subseteq[k]:|S|\le d} ,\; f_S(x) = 0
\right\}.
\]
Hence, $A$ is the set of solutions to a system of $\binom{k}{\le d}$ polynomial equations,
where there are $\binom{k}{j}$ equations which correspond to sets $S$ of size $j$ and thus to degree (at most) $d-j$ polynomials.~\footnote{In particular, equations that correspond to sets $S$ of size $d$ are of the form $c_S = 0$ for some constant $c_S\in \F_2$. Since $A$ is non-empty, the constants $c_S$ must be $0$, making those equations tautologies $0=0$ that does not depend on $x$. Moreover, most of the remaining equations correspond to sets $S$ of size $d-1$, and are therefore either linear equations or tautologies.}
One can also write $A$ as the set of solutions to the single polynomial equation
\[
\prod_{S\subseteq[k]:|S|\le d}{\!\!\!\!\!\!(1-f_S(x))} \,\,= 1,
\]
which is of degree
\[
D \le \sum_{j=0}^{d-1} {(d-j) \cdot \binom{k}{j}}\;.
\]
Since $A$ is non-empty, by DeMillo-Lipton-Schwartz-Zippel lemma (Lemma~\ref{lem:dlsz}, for $q=2$) we have that
\begin{equation}
|A| \ge 2^{n-D} \ge 2^{n-\sum_{j=0}^{d-1} {(d-j) \cdot \binom{k}{j}}}.
\end{equation}
This, together with Equation \eqref{eq:condition on k} implies that $|A|>2^{k}$.
Hence, there exists a point $y\in A$ such that $y - u_0 \notin \spa\{\Delta_1, \Delta_2, \ldots, \Delta_{k}\}$.
Pick such a point $u$ arbitrarily and denote by $\Delta_{k+1} \deff u - u_0$.
Since both $u_0$ and $u$ are in $A$ we have that $f\equiv 0$  on
\[
u_0 + \spa \{\Delta_1, \ldots, \Delta_{k+1}\} \;.
\]
The inductive proof shows that there exists a subspace $U$ of dimension $k$ such that $f$ is constant on $u_0 + U$ and
\begin{equation}\label{eq:bound on k}
n \le k + \sum_{j=0}^{d-1} {(d-j) \cdot \binom{k}{j}} \;,
\end{equation}
since otherwise we could have continue this process and pick a bigger subspace $U'$.
\remove{We now complete the proof by showing that $k = \Omega(d \cdot n^{1/(d-1)})$ for $d \le \log(n)/3$ and $k = \Omega(n^{1/(d-1)})$ for any $d$.
The right hand side of Equation~\eqref{eq:bound on k} is bounded above by $d \cdot 2^k$,
hence $k \ge \log(n/d)$. Under the assumption $d \le \log(n)/3$ we get $k \ge 2\log(n)/3 \ge 2d$.
We return to Equation~\eqref{eq:bound on k} and deduce that
$$
n \le k + d \cdot \sum_{j=0}^{d-1}{\binom{k}{j}}  \underset{2d \le k}{\le} (d^2+1) \cdot \binom{k}{d-1} \le (d^2+1) \cdot \left (\frac{k e}{d-1} \right)^{d-1}
$$
and so
$$
k \ge \left( \frac{n}{d^2+1} \right)^{\frac1{d-1}} \cdot \frac{d-1}{e} > \frac{1}{28} \cdot d \cdot n^{1/(d-1)}\;.
$$
For $d \ge  \log(n)/3$ the proof follows since $n^{1/(d-1)} \le 64$.
}
\end{proof}

\subsection{Proof of Theorem~\ref{thm:efficient structural result I}}\label{sec:binary field algorithmic}

The proof of Theorem~\ref{thm:efficient structural result I} uses the following lemma.
\begin{lemma}\label{lemma:degree decreases equivalence}
Let $f: \F_2^n \to \F_2$ be a degree $d$ polynomial, and let $U$ be a linear subspace with basis $\Delta_1, \ldots,\Delta_k$. Then, $\deg(f|_U) \le d-1$ if and only if $f_S(0) = 0$ for all $S\subseteq[k]$ of size $d$, where $f_S(x) := \sum_{T\subseteq S} {f\(x + \sum_{i \in T}{\Delta_i}\)}$.
\end{lemma}
\begin{proof}[Proof of Lemma~\ref{lemma:degree decreases equivalence}]
As noted in the Preliminaries section, the degree of $f|_U$ is equal to the degree of $g: \F_2^k \to \F_2$ defined as $g(y_1, \ldots, y_k) = f(\sum_{i=1}^{k}{y_i \Delta_i})$.
Since $\deg(g) \le d$, we may write $g(y) = \sum_{S \subseteq [k],|S|\le d} {a_S \cdot \prod_{i\in S}{y_i}}$, where $a_S \in \F_2$ are constants.
By M\"{o}bius inversion formula (Fact~\ref{fact:mobius}), $a_S = \sum_{T \subseteq S} {g(\one_T )}$.
By the definition of $g$, we establish the relation $a_S = \sum_{T\subseteq S}{ f(\sum_{i\in T}{\Delta_i})} = f_S(0)$.
Hence,
\begin{align*}
\deg(f|_U) \le d-1 &\quad \iff \quad \deg(g) \le d-1 \\
&\quad \iff \quad \forall{S\subseteq[k] \text{\quad s.t. } |S|=d}, a_S = 0 \\
&\quad \iff \quad \forall{S\subseteq[k] \text{\quad s.t. } |S|=d}, f_S(0) = 0,
\end{align*}
which completes the proof.
\end{proof}
\begin{proof}[Proof of Theorem~\ref{thm:efficient structural result I}]
Similarly to the proof of Theorem~\ref{thm:structural result binary field}, we find by induction basis vectors $\Delta_1, \ldots, \Delta_k$ for the subspace $U$.
We assume by induction that $\deg(f|_U) \le d-1$, and we wish to find a new vector $\Delta_{k+1}$, linearly independent of $\Delta_1, \ldots, \Delta_k$, for which $\deg(f|_{U'}) \le d-1$, where $U' = \spa\{\Delta_1, \ldots, \Delta_{k+1}\}$. We continue doing so as long as
$\binom{k}{d-1} + k < n$.\footnote{Note that this is slightly better than the expression we had in Theorem~\ref{thm:structural result binary field}.}

By Lemma~\ref{lemma:degree decreases equivalence}, for any set $S \subseteq [k]$ of size $d$, $f_S(0) = 0$.
We wish to find a new vector $\Delta_{k+1}$ such that for all $S\subseteq[k+1]$ of size $d$, $f_S(0) = 0$.
It suffices to consider sets $S$ of size $d$ that contains $k+1$, since the correctness for all other sets is implied by the induction hypothesis.

For sets $S$ of size $d-1$, $f_S(x)$ is an affine function and can be written as $f_S(x) = \langle\ell_S,x \rangle +c_S$, where $\ell_S \in \F_2^n$ and $c_S \in \F_2$.
Let $W$ be the linear subspace of $\F_2^n$ spanned by $\{\ell_S: S\subseteq[k], |S|=d-1\}$.
Let $\Delta_{k+1}$ be any vector orthogonal to $W$, and linearly independent of $\Delta_1, \Delta_2, \ldots, \Delta_k$. Since, $\dim(W^{\perp}) = n - \dim(W) \ge n-\binom{k}{d-1}$, which by our assumption is strictly bigger than $k$, such a vector $\Delta_{k+1}$ exists.
Let $S\subseteq[k+1]$ be a set of size $d$ that contains $k+1$ and let $S' = S \cap [k]$, then
\begin{align*}
f_{S}(0) = f_{S'}(0) + f_{S'}(\Delta_{k+1})
= \langle \ell_{S'}, 0 \rangle + c_{S'} + \langle\ell_{S'},\Delta_{k+1}\rangle + c_{S'}
= 0\;,
\end{align*}
where in the first equality we used the definitions of $f_S$ and $f_{S'}$, and in the last equality we used the fact that $\Delta_{k+1}$ is orthogonal to $\ell_{S'}$.
Using Lemma~\ref{lemma:degree decreases equivalence} we have shown that our choice of $\Delta_{k+1}$ gives a linear subspace $U' = \spa\{\Delta_1, \ldots, \Delta_{k+1}\}$ for which $f|_{U'}$ is of degree $\le d-1$.

We now explain how to find, for any set $S$ of size $d-1$, the affine function $f_S(x)$ (that is, $\ell_S$ and $c_S$) by performing $2^{d-1}\cdot (n+1)$ queries to $f$. As $f_S$ is affine, knowing the values of $f_S$ on the inputs $0,e_1,e_2, \ldots, e_n$ determines $\ell_S$ and $c_S$: $c_S = f_S(0)$ and $(\ell_S)_i = c_S + f_S(e_i)$ for $i\in [n]$. Each one of the values $f_S(0), f_S(e_1), \ldots, f_S(e_n)$ can be computed using $2^{d-1}$ queries to $f$, by the definition of $f_S$. 

We now describe how can one efficiently find the vector $\Delta_{k+1}$ given $\Delta_1, \ldots, \Delta_k$. Using Gaussian elimination we find a basis for $W^{\perp}$. We check for each basis vector if it is not in the span of $\Delta_1, \ldots, \Delta_k$; after checking $k+1$ vectors we are promised to find such a vector. Next, we analyze the dimension of the subspace returned by the algorithm, the number of queries it makes to $f$, and the total running time.

\paragraph{Dimension of subspace:} We abuse notation and denote by $k$ the number of rounds in our algorithm, which is also the dimension of the subspace the algorithm returns. Since the algorithm stopped, we know that $\binom{k}{d-1} + k \ge n$. By a simple calculation, under the assumption that $d\le \log(n)/3$ we get that
$k = \Theta(d \cdot n^{1/(d-1)})$.

\paragraph{Number of queries:} Overall through the $k$ rounds of the algorithm we query $f$ on all vectors of the form $v  + \sum_{i \in T}{\Delta_i}$ for $v\in\{0,e_1,\ldots,e_n\}$ and $T \subseteq [k]$ of size $\le d-1$. Hence, if we make sure not to query $f$ more than once on the same point, the number of queries is $(n+1) \cdot \binom{k}{\le d-1}$ which is at most $O(n^2)$ for $d \le \log(n)/3$.

\paragraph{Running time:}The total running time per round is $O(n^3)$ since we perform Gaussian elimination to calculate the basis for $W^{\perp}$, and another Gaussian elimination to check which of the first $k+1$ vectors of this basis is not in $\spa\{\Delta_1, \ldots, \Delta_{k+1}\}$.
In addition, in each round we calculate the linear functions $\ell_S$, but this only takes $O(n^2 \cdot 2^d)$ time, which is negligible compared to $O(n^3)$ under the assumption that $d\le \log(n)/3$.
Therefore, the total running time is $O(n^3 \cdot k)$.
\end{proof}

\section{Proof of DeMillo-Lipton-Schwartz-Zippel Variant}\label{sec:dlsz}

In this section we provide a proof for Lemma~\ref{lem:dlsz}. Our proof is adapted from the proof of Lemma A.36 in the book of Arora and Barak \cite{AB09}.
\begin{proof}[Proof of Lemma~\ref{lem:dlsz}]
Since we only care about the values the polynomial take on $\F_q^n$, we may assume without loss of generality that the individual degree of each variable is at most $q-1$, since $a^q = a$ for all $ a\in \F_q$.

We use induction on $n$. If $n = 1$ then $f$ is a univariate polynomial of degree $d$ for some $d\le q-1$, since we assumed each individual degree is at most $q-1$. We have
\[
\Pr[f(x_1) \neq 0] \ge 1-d/q \ge q^{-d/(q-1)},
\]
where the first inequality follows since a univariate degree $d$ polynomial over a field obtains at most $d$ roots, and the last inequality can be verified for any $d \le q-1$ using basic calculus.
Suppose the statement is true when the number of variables is at most $n-1$. Then $f$ can be written as
\[
f(x_1, \ldots, x_n)  = \sum_{i=0}^{\min(d,q-1)} {x_1^i \cdot f_i(x_2, \ldots, x_n)}
\]
where $f_i$ is of total degree at most $d-i$. Let $k$ be the largest $i$ such that $f_i$ is a non-zero polynomial.
By conditioning we have,
\[
\Pr[f(x_1, \ldots, x_n) \neq 0] \ge \Pr[f_k(x_2, \ldots, x_n) \neq 0] \cdot \Pr[f(x_1, \ldots, x_n) \neq 0 \mid f_k(x_2, \ldots, x_n) \neq 0]\;.
\]
By the induction hypothesis, the first multiplicand is at least $q^{-(d-k)/(q-1)}$.
As for the second multiplicand, for any fixed $(x_2, \ldots, x_n) = (a_2, \ldots, a_n)$ such that $f_k(a_2, \ldots, a_n)\neq 0$, we get that $f(x_1, a_2, \ldots, a_n)$ is a non-zero univariate polynomial, in the variable $x_1$, of degree $k$. Hence, $\Pr_{x_1 \sim \F_q}[f(x_1, a_2, \ldots, a_n) \neq 0] \ge q^{-k/(q-1)}$ from the base case. Overall we get
\[
\Pr[f(x_1, \ldots, x_n) \neq 0] \ge q^{-(d-k)/(q-1)} q^{-k/(q-1)} = q^{-d/(q-1)}\;.
\]
\end{proof} 
\end{document}